\documentclass[sigconf]{acmart}

\AtBeginDocument{%
  \providecommand\BibTeX{{%
    \normalfont B\kern-0.5em{\scshape i\kern-0.25em b}\kern-0.8em\TeX}}}

\setcopyright{acmcopyright}
\copyrightyear{2018}
\acmYear{2018}
\acmDOI{XXXXXXX.XXXXXXX}

\acmConference[Conference acronym 'XX]{Make sure to enter the correct
  conference title from your rights confirmation emai}{June 03--05,
  2018}{Woodstock, NY}
\acmPrice{15.00}
\acmISBN{978-1-4503-XXXX-X/18/06}

\usepackage[linesnumbered, ruled]{algorithm2e}
\usepackage[noend]{algpseudocode}

\SetCommentSty{mycommfont}
\SetKwInput{KwInput}{Input}                %
\SetKwInput{KwOutput}{Output}              %

\usepackage{booktabs}
\usepackage{multirow}

\usepackage{subfig}
\usepackage{graphicx}
\usepackage{stfloats}

\usepackage[skip=6pt]{caption}

\newtheorem{definition}{Definition}

\newcounter{example}
\newenvironment{example}[1][]{\refstepcounter{example}\par\medskip   \noindent \textbf{Example~\theexample. #1} \rmfamily}{\qedsymbol  \smallskip}

\let\emptyset\varnothing

\newcommand{\mc}[1]{\mathcal{#1}}

\usepackage{xspace}

\newcommand{\oursystem}{\texttt{DProvDB}\xspace}

\newcommand{\provT}{privacy provenance table\xspace}
\newcommand{\mypaper}{paper\xspace}
\newcommand{\mysection}{section\xspace}
\newcommand{\mySection}{Section\xspace}

\newcommand{\baseline}{vanilla\xspace}
\newcommand{\Baseline}{Vanilla\xspace}

\newcommand{\admin}{administrator\xspace}
\newcommand{\Admin}{Administrator\xspace}

\usepackage{amsmath}

\DeclareMathOperator*{\argmax}{arg\,max}

\newcommand{\databasedom}{\mc{D}}

\newcommand{\xh}[1]{{\color{purple} (XH: #1)}}

\newcommand{\eat}[1]{}

\newcommand{\revise}[1]{{#1}}
\newcommand{\shepherding}[1]{{#1}}

\newcommand{\stitle}[1]{\smallskip \noindent{\bf #1}}

\newtheorem{problem}{Problem}

\newcommand{\squishlist}{
	\begin{list}{$\bullet$}
		{
			\setlength{\itemsep}{0pt}
			\setlength{\parsep}{3pt}
			\setlength{\topsep}{3pt}
			\setlength{\partopsep}{0pt}
			\setlength{\leftmargin}{1.5em}
			\setlength{\labelwidth}{1em}
			\setlength{\labelsep}{0.5em} } }

\newcommand{\squishend}{
	\end{list}  }

\usepackage{enumitem}

\DeclareMathOperator\e{e}

\usepackage{circledtext}
\usepackage{pgffor}
\circledtextset{boxcolor=black,boxtype=ox,charf=\LARGE,boxfill = red!20}

\newcommand{\linelabel}[1]{}

\renewcommand\footnotetextcopyrightpermission[1]{} %

\newif\ifarxiv
\arxivtrue     %

\newif\ifconf
\conffalse   %

\begin{document}

\title{DProvDB: Differentially Private Query Processing with Multi-Analyst Provenance}
\titlenote{This is the full version of the work being accepted for publication in \emph{Proc. of the ACM on Management of Data} (SIGMOD 2024).}

\author{Shufan Zhang}
\affiliation{%
  \institution{University of Waterloo}
  \country{Waterloo, Canada}
}
\email{shufan.zhang@uwaterloo.ca}

\author{Xi He}
\affiliation{%
  \institution{University of Waterloo}
  \country{Waterloo, Canada}
}
\email{xi.he@uwaterloo.ca}

\begin{abstract}
Recent years have witnessed the adoption of differential privacy (DP) in practical database systems like PINQ, FLEX, and PrivateSQL. Such systems allow data analysts to query sensitive data while providing a rigorous and provable privacy guarantee.
However, the existing design of these systems does not distinguish data analysts of different privilege levels or trust levels. This design can have an unfair apportion of the privacy budget among the data analyst if treating them as a single entity, or waste the privacy budget if considering them as non-colluding parties and answering their queries independently. In this \mypaper, we propose \oursystem, a fine-grained privacy provenance framework for the multi-analyst scenario that tracks the privacy loss to each single data analyst. Under this framework, when given a fixed privacy budget, we build algorithms that maximize the number of queries that could be answered accurately and apportion the privacy budget according to the privilege levels of the data analysts.
\end{abstract}

\maketitle

\setcounter{page}{1}

\pagestyle{plain}

\section{Introduction}
\label{sec:introduction}

With the growing attention on data privacy and the development of privacy regulations like GDPR \cite{gdpr}, companies with sensitive data must share their data without compromising the privacy of data contributors. Differential privacy (DP)~\cite{dwork2014algorithmic} has been considered as a promising standard for this setting. Recent years have witnessed the adoption of DP to practical systems for data management and online query processing, such as PINQ \cite{mcsherry2009privacy}, FLEX \cite{johnson2018towards}, PrivateSQL \cite{kotsogiannis2019privatesql}, GoogleDP \cite{amin2022plume}, and Chorus \cite{johnson2020chorus}. In systems of this kind, data curators or system providers set up a finite system-wise privacy budget to bound the overall extent of information disclosure. An incoming query consumes some privacy budget. The system stops processing new queries once the budget has been fully depleted. Thus, the privacy budget is a crucial resource to manage in such a query processing system. 

In practice, multiple data analysts can be interested in the same data, and they have different privilege/trust levels in accessing the data. For instance, 
tech companies need to query their users' data for internal applications like anomaly detection. They also consider inviting  external researchers with low privilege/trust levels to access the same sensitive data for study. Existing query processing systems with DP guarantees would regard these data analysts as a unified entity and do not provide tools to distinguish them or track their perspective privacy loss. 
This leads to a few potential problems. First, a low-privilege external data analyst who asks queries first can consume more privacy budget than an internal one, if the system does not interfere with the sequence of queries. Second, if na\"ively tracking and answering each analyst's queries independently of the others, the system can waste the privacy budget when two data analysts ask similar queries. 

The aforementioned challenges to private data management and analytics are mainly on account of the fact that the systems are ``\emph{stateless}'', meaning none of the existing DP query-processing systems records the individual budget limit and the historical queries asked by the data analysts.
That is, the \textbf{metadata} about \emph{where the query comes from, how the query is computed, and how many times each result is produced}, which is related to the \textbf{provenance information} in database research \cite{buneman2007provenance,cheney2009provenance}. As one can see, without privacy provenance, the query answering process for the multi-analyst use case can be unfair or wasteful in budget allocation.

To tackle these challenges, we propose a ``stateful'' DP query processing system \oursystem, which enables a novel privacy provenance framework designed for the multi-analyst setting.
Following the existing work \cite{kotsogiannis2019architecting}, \oursystem answers queries based on private synopses (i.e., materialized results for views) of data.
Instead of recording all the query metadata, we propose a more succinct data structure --- a \provT, that enforces only necessary privacy tracing as per each data analyst and per view.
The \provT is associated with privacy constraints so that constraint-violating queries will be rejected.
Making use of this privacy provenance framework, \oursystem can maintain global (viz., as per view) and local (viz., as per analyst) DP synopses and update them dynamically according to data analysts' requests.

\revise{\linelabel{line:intro}\oursystem is supported with a new principled method, called \emph{additive Gaussian approach}, to manage DP synopses.
The additive Gaussian approach leverages DP mechanism that adds correlated Gaussian noise to mediating unnecessary budget consumption across data analysts and over time.}
This approach first creates a global DP synopsis for a view query; Then, from this global synopsis, it provides the necessary local DP synopsis to data analysts who are interested in this view by adding more Gaussian noise. In such a way \oursystem is tuned to answer as many queries accurately from different data analysts. Even when all the analysts collude, the privacy loss will be bounded by the budget used for the global synopsis. 
Adding up to its merits, we notice the provenance tracking in \oursystem can help in achieving a notion called proportional fairness.
We believe most of if not all, existing DP query processing systems can benefit from integrating our multi-analyst privacy provenance framework ---
\oursystem can be regarded as a middle-ground approach between the purely interactive DP systems and those based solely on synopses, from which we both provably and empirically show that \oursystem can significantly improve on system utility and fairness for multi-analyst DP query processing.

The contributions of this paper are the following:
\squishlist
    \item We propose a multi-analyst DP model where mechanisms satisfying this DP provide discrepant answers to analysts with different privilege levels. Under this setting, we ask research questions about tight privacy analysis, budget allocation, and fair query processing. \revise{(Section $\S$\ref{sec:problem_setup})}
    
    \item 
    We propose a privacy provenance framework that compactly traces historical queries and privacy consumption as per analyst and as per view. 
    With this framework, the \admin is able to enforce privacy constraints, enabling dynamic budget allocation and fair query processing. \revise{(Section $\S$\ref{sec:system_overview})}
    
    \item We design new accuracy-aware DP mechanisms that leverages the provenance data to manage synopses and inject correlated noise to achieve tight collusion bounds \revise{over time} in the multi-analyst setting. The proposed mechanisms can be seamlessly added to the algorithmic toolbox for DP systems. \revise{(Section $\S$\ref{sec:low_sensitive_algos})}

    \item 
    We implement \oursystem\footnote{The system code is available at \url{https://github.com/DProvDB/DProvDB}.}, as a new multi-analyst query processing interface, and integrate it into an existing DP query system. We empirically evaluate \oursystem, and the experimental results show that our system is efficient and effective compared to baseline systems. \revise{(Section $\S$\ref{sec:evaluation})}
\squishend

\ifarxiv
\revise{
\stitle{Paper Roadmap.}\linelabel{line:roadmap}
The remainder of this \mypaper is outlined as follows.
\mySection \ref{sec:preliminaries} introduces the necessary notations and background knowledge on database and DP.
Our multi-analyst DP query processing research problems are formulated in \mysection \ref{sec:problem_setup} and a high-level overview of our proposed system is briefed in \mysection \ref{sec:system_overview}.
\mySection \ref{sec:low_sensitive_algos} describes the details of our design of the DP mechanisms and system modules.
In \mysection \ref{sec:evaluation}, we present the implementation details and an empirical evaluation of our system against the baseline solutions.
Section \ref{sec:discussion_new} discusses extensions to the compromisation model and other strawman designs.
In \mysection \ref{sec:related_work} we go through the literature that is related and we conclude this work in \mysection \ref{sec:conclusion}.
} 
\fi

\section{Preliminaries}
\label{sec:preliminaries}

Let $\databasedom$ denotes the domain of database and $D$ be a database instance.
A relation $R \in D$ consists of a set of attributes, $attr({R}) = \{ a_1, \dots, a_j \}$.
We denote the domain of an attribute $a_j$ by $Dom(a_j)$ while $|Dom(a_j)|$ denotes the domain size of that attribute.
We introduce and summarize the related definitions of differential privacy.

\begin{definition}[Differential Privacy \cite{dwork2006calibrating}]

We say that a randomized algorithm $\mc{M}: \databasedom \rightarrow  \mc{O}$ 
satisfies $(\epsilon, \delta)$-differential privacy (DP), if for any two
neighbouring databases $(D,D')$ that differ in only 1 tuple, 
and $O\subseteq \mc{O}$, we have
$$
\Pr[\mc{M}(D) \in O] \leq e^\epsilon \Pr[\mc{M}(D') \in O] + \delta.
$$
\end{definition}

\noindent
DP enjoys many useful properties, for example, post-processing and sequential composition \cite{dwork2014algorithmic}.

\begin{theorem}[Sequential Composition \cite{dwork2014algorithmic}]
\label{thm:dp_seq_composition}
Given two mechanisms $\mc{M}_1: \databasedom \rightarrow  \mc{O}_1$ and $\mc{M}_2: \databasedom \rightarrow  \mc{O}_2$, such that $\mc{M}_1$ satisfies $(\epsilon_1, \delta_1)$-DP and $\mc{M}_2$ satisfies $(\epsilon_2, \delta_2)$-DP.
The combination of the two mechanisms $\mc{M}_{1, 2}: \databasedom \rightarrow  \mc{O}_1 \times \mc{O}_2$, which is a mapping $\mc{M}_{1, 2}(D) = (\mc{M}_1(D), \mc{M}_2(D))$, is $(\epsilon_1+\epsilon_2,\delta_1+\delta_2)$-DP.
\end{theorem}

\begin{theorem}[Post Processing \cite{dwork2014algorithmic}]
\label{thm:dp_post_processing}
For an $(\epsilon, \delta)$-DP mechanism $\mc{M}$, applying any arbitrary function $f$ over the output of $\mc{M}$, that is, the composed mechanism $f \circ \mc{M}$, satisfies $(\epsilon, \delta)$-DP.
\end{theorem}
The sequential composition bounds the privacy loss of the sequential execution of DP mechanisms over the database. It can naturally be generalized to the composition of $k$ differentially private mechanisms.
The post-processing property of DP indicates that the execution of any function on the output of a DP mechanism will not incur privacy loss.

\begin{definition}[$\ell_2$-Global Sensitivity]
For a query $q: \mathcal{D} \rightarrow \mathbb{R}^d$ %
the $\ell_2$ global sensitivity of this query is 
$$
\Delta q = \max_{D, D': d(D,D') \leq 1} \| q(D) - q(D') \|_2,
$$ 
where $d(\cdot, \cdot)$ denotes the number of tuples that $D$ and $D'$ differ and $\| \cdot \|_2$ denotes the $\ell_2$ norm.
\end{definition}

\begin{definition}[Analytic Gaussian Mechanism \cite{BW18analytic}]
\label{def:analytic_gaussian}
Given a query $q : \mathcal{D} \rightarrow \mathbb{R}^{d}$, the analytic Gaussian mechanism $\mathcal{M}(D) = q(D) + \eta$  where 
$\eta \sim {\mc{N}} \left( 0, \sigma^2 I \right)$ 
is $(\epsilon, \delta)$-DP if and only if 
$$
\Phi_{\mc{N}} \left( \frac{\Delta q}{2 \sigma} - \frac{\epsilon \sigma}{\Delta q} \right)
    - e^{\epsilon} \Phi_{\mc{N}} \left( - \frac{\Delta q}{2 \sigma} - \frac{\epsilon \sigma}{\Delta q} \right) 
    \leq \delta,
$$
where $\Phi_{\mc{N}}$ denotes the cumulative density function (CDF) of Gaussian distribution.
\end{definition}
In this mechanism, the Gaussian variance is determined by $\sigma = \alpha \Delta q / \sqrt{2\epsilon}$ where $\alpha$ is a parameter determined by $\epsilon$ and $\delta$~\cite{BW18analytic}. 

DP mechanisms involve errors in the query answer. A common \revise{and popular} utility function for numerical query answers is the expected mean squared error (MSE), defined as follows:

\begin{definition}[Query Utility]
\label{def:expected_err}
For a query $q: \mathcal{D} \rightarrow \mathbb{R}^d$ and a mechanism $\mc{M}: \mathcal{D} \rightarrow \mathbb{R}^d$, the \textit{query utility} of mechanism $\mc{M}$ is measured as the expected squared error, $v = \mathbb{E}[q(D) - \mc{M}(D)]^2$ $=  \sum_{i=1}^n (q(D)[i] - \mc{M}(D)[i])^2$.
For the (analytic) Gaussian mechanism, the expected squared error equals its variance, that is, $v = d\sigma^2$. 
\end{definition}

\section{Problem Setup}
\label{sec:problem_setup}

We  consider \revise{a new} \textbf{multi-analyst} setting, where there are multiple data analysts $\mc{A} = \{ A_1, \dots, A_m \}$ who want to ask queries on the database $D$. 
The \admin who manages the database wants to ensure that the sensitive data is properly and privately shared with the data analysts $A_1, \dots, A_m$. 
In our threat model, the data analysts can
adaptively select and submit arbitrary queries to the system to infer sensitive information about individuals in the protected database.
Additionally, in our multi-analyst model, data analysts may submit the same query and collude to learn more information about the sensitive data.

In practice, data analysts have varying privilege levels, which the standard DP  does not consider. An analyst with a higher privilege level is allowed more information than an analyst with a lower privilege level. This  motivates us to present the following variant of DP that supports the privacy per analyst provenance framework and guarantees different levels of privacy loss to the analysts.

\begin{definition}[Multi-analyst DP]\label{def:multi_analyst_dp}
We say a randomized mechanism $\mc{M}: \mc{D}  \rightarrow (\mc{O}_1, \ldots, \mc{O}_m)$ satisfies $[(A_1, \epsilon_1, \delta_1), ..., (A_m, \epsilon_m, \delta_m)]$-multi-analyst-DP if for any neighbouring databases $(D,D')$, any $i \in [m]$, and all $O_i \subseteq \mc{O}_i$, we have
$$
\Pr[\mc{M}(D) \in O_i] \leq e^{\epsilon_i} \Pr[\mc{M}(D') \in O_i] + \delta_i,
$$
where $O_i$ are the outputs
released to the $i$th analyst.
\end{definition}

The multi-analyst DP variant supports the composition across different algorithms, as indicated by the following theorem.

\begin{theorem}[Multi-Analyst DP Composition]
\label{thm:sequential_composition_multi_dp}
Given two randomized mechanisms $\mc{M}_1$ and $\mc{M}_2$, where $\mc{M}_1: \mc{D} \rightarrow (\mc{O}_1, \ldots, \mc{O}_m)$ satisfies $[(A_1, \epsilon_1, \delta_1), ..., (A_m, \epsilon_m, \delta_m)]$-multi-analyst-DP, and $\mc{M}_2: \mc{D} \rightarrow (\mc{O}_1', \ldots, \mc{O}_m')$ satisfies $[(A_1, \epsilon_1', \delta_1'), ..., (A_m, \epsilon_m', \delta_m')]$-multi-analyst-DP, then the mechanism $g (\mc{M}_1, \mc{M}_2)$
gives the $[(A_1, \epsilon_1+\epsilon_1', \delta_1+\delta_1'), ..., (A_m, \epsilon_m+\epsilon_m', \delta_m+\delta_m')]$-multi-analyst-DP guarantee.
\end{theorem}

\revise{
\linelabel{line:setup_1}
Unlike prior DP work for multiple data analysts, our setup considers data analysts who are obliged under laws/regulations should not share their privacy budget/query responses with each other. We provide a detailed discussion on comparing other work in Section~\ref{sec:related_work}.

Under our new multi-analyst DP framework, several, natural but less well-understood, research questions (RQs) are raised 
and problem setups are considered of potential interest. 
}

\stitle{RQ 1: worst-case privacy analysis across analysts.} Under this multi-analyst DP framework, what if all data analysts collude or are compromised by an adversary, how could we design algorithms to account for the privacy loss to the colluded analysts?
When this happens, we can obtain the following trivial lower bound and upper bound for the standard DP measure. 

\begin{theorem}[Compromisation Lower/Upper Bound, Trivial]
    Given a mechanism $\mc{M}$ that satisfies $[(A_1, \epsilon_1, \delta_1), ..., $ $(A_n, \epsilon_n, \delta_n)]$-multi-analyst-DP, when all data analysts collude, its DP loss is (i) lowered bounded by $(\max \epsilon_i, \max \delta_i)$, where $(\epsilon_i,\delta_i)$ is the privacy loss to the $i$-th analyst, and (ii) trivially upper bounded by $(\sum \epsilon_i, \sum \delta_i)$.
\end{theorem}

\ifarxiv
The lower bound indicates the least amount of information that has to be released (to the analyst) and this upper bound is simply derived from sequential composition.
\fi
Obviously, the trivial upper bound does not match the lower bound, rendering the question of designing multi-analyst DP mechanisms to close the gap. 
Closing this gap means \emph{even if these data analysts break the law and collude, the overall privacy loss of the multi-analyst DP mechanism is still minimized.}
In this paper, we would like to design such algorithms that achieve the lower bound, as shown in Section \ref{sec:low_sensitive_algos}.

\stitle{RQ 2: dynamic budget allocation across views. } 
The DP query processing system should impose constraints on the total privacy loss by all the analysts (in the worst case) and the privacy loss per analyst.
When handling incoming queries, prior work either dynamically allocates the privacy budget based on the budget request per query \cite{mcsherry2009privacy,johnson2018towards} or query accuracy requirements  \cite{ge2019apex} or predefine a set of static DP views that can handle the incoming queries \cite{kotsogiannis2019privatesql}.

The dynamic budget per query approach can deplete the privacy budget quickly as each query is handled with a fresh new privacy budget. The static views spend the budget in advance among them to handle unlimited queries, but they may fail to meet the accuracy requirements of some future queries. Therefore, in our work, we consider the view approach but assign budgets dynamically to the views based on the incoming queries so that more queries can be answered with their accuracy requirements. Specifically, we would like to use the histogram view, which queries the number of tuples in a database for each possible value of a set of attributes. The answer to a view is called a synopsis. We consider a set of views that can answer all incoming queries.

\begin{definition}[Query Answerability \cite{kotsogiannis2019privatesql}]
For a query $q$ over the database $D$, if there exists a query $q'$ over the histogram view $V$ such that $q(D) = q(V(D))$, we say that $q$ is answerable over $V$.
\end{definition}

\begin{example}
Consider two queries $q_1$ and $q_2$ over a database for employees in Figure~\ref{fig:provenance_table}. They are answerable over the $V_1$, a 3-way marginal contingency table over attributes (age, gender, education), via their respective transformed queries $\hat{q}_1$ and $\hat{q}_2$.
\end{example}

Given a set of views, we would like to design algorithms that can dynamically allocate privacy budgets to them and update their corresponding DP synopses over time. We show how these algorithms maximize the queries that can be handled accurately in Section \ref{sec:low_sensitive_algos}.
Since we can dynamically allocate budget to views, our system can add new views to the system as overtime. We discuss this possibility in Section \ref{subsec:algo_col_constraints}.

\stitle{RQ 3: fair query answering among data analysts.} 
A fair system expects data analysts with higher privacy privileges to receive more accurate answers or larger privacy budgets than ones with lower privacy privileges. However, existing DP systems make no distinctions among data analysts. Hence, it is possible that a low-privilege external analyst who asks queries first consumes all the privacy budget and receives more accurate query answers, leaving no privacy budgets for high-privilege internal data analysts. It is also impossible to force data analysts to ask queries in a certain order. In this context, we would like the system to set up the (available and consumed) privacy budgets for data analysts according to their privacy privilege level. In particular, we define privacy privilege levels as an integer in the range of 1 to 10, where a higher number represents a higher privilege level. We also define a fairness notion inspired by the literature on resource allocation \cite{kelly1998rate,banerjee2022proportionally,stoyanovich2018online}.

\begin{definition}[Proportional Fairness]
    \label{def:prop_fairness}
    Consider a DP system handling a sequence of queries $Q$ 
    from multiple data analysts with a mechanism $\mc{M}$, where each data analyst $A_i$ is associated with a privilege level $l_i$.   
    We say the mechanism $\mc{M}$ satisfies proportional fairness, if $\forall A_i, A_j$ ($i\neq j$), $l_i \leq l_j$, we have
$$
\frac{Err_i(M, A_i, Q)}{\mu(l_i)} \leq \frac{Err_j(M, A_j, Q)}{\mu(l_j)},
$$
where $Err_i(M, A_i, Q)$ denotes the analyst $A_i$'s privacy budget consumption and $\mu(\cdot)$ is some linear function. 
\end{definition}

This fairness notion suggests the quality of query answers to data analysts, denoted by $Err_i(M, A_i, Q)$ is proportional to their privilege levels, denoted by a linear function $\mu(\cdot)$ of their privilege levels. We first consider the privacy budget per data analyst as the quality function, as a smaller error to the query answer is expected with a larger privacy budget. 
We show in Section \ref{subsec:algo_col_constraints} how to set up the system to achieve fairness when the analysts ask a sufficient number of queries, which means they finish consuming their assigned privacy budget.

\section{System overview}

\label{sec:system_overview}

In this \mysection, we outline the key design principles of \oursystem and briefly describe the modules of the system.

\subsection{Key Design Principles}

To support the multi-analyst use case and to answer the aforementioned research questions, we identify the four principles %
and propose a system \oursystem that follows these principles.

\stitle{Principle 1: fine-grained privacy provenance.} 
The query processing system should be able to track the privacy budget allocated per each data analyst and per each view in a fine-grained way.
The system should additionally enable a mechanism to compose privacy loss across data analysts and the queries they ask.

\stitle{Principle 2: view-based privacy management.}
The queries are answered based on DP views or synopses. %
Compared to directly answering a query from the database $D$, view-based query answering can answer more private queries \cite{kotsogiannis2019privatesql}, but it assumes the accessibility of a pre-known query workload.
In our system, the view is the minimum data object that we keep track of its privacy loss, and the views can be updated dynamically if higher data utility is required.
The privacy budgets spent on different views during the updating process depend on the incoming queries.

\stitle{Principle 3: dual query submission mode.}
Besides allowing data analysts to submit a budget with their query, the system enables an accuracy-aware mode.
With this mode, data analysts can submit the query with their desired accuracy levels regarding the expected squared error.
The dual mode system supports data analysts from domain experts, who can take full advantage of the budgets, to DP novices, who only care about the accuracy bounds of the query.

\stitle{Principle 4: maximum query answering.}
The system should be tuned to answer as many queries accurately as possible without violating the privacy constraint specified by the \admin as per data analyst and per view based on their privilege levels.

\begin{figure}[t]
    \centering
    \includegraphics[width=\linewidth]{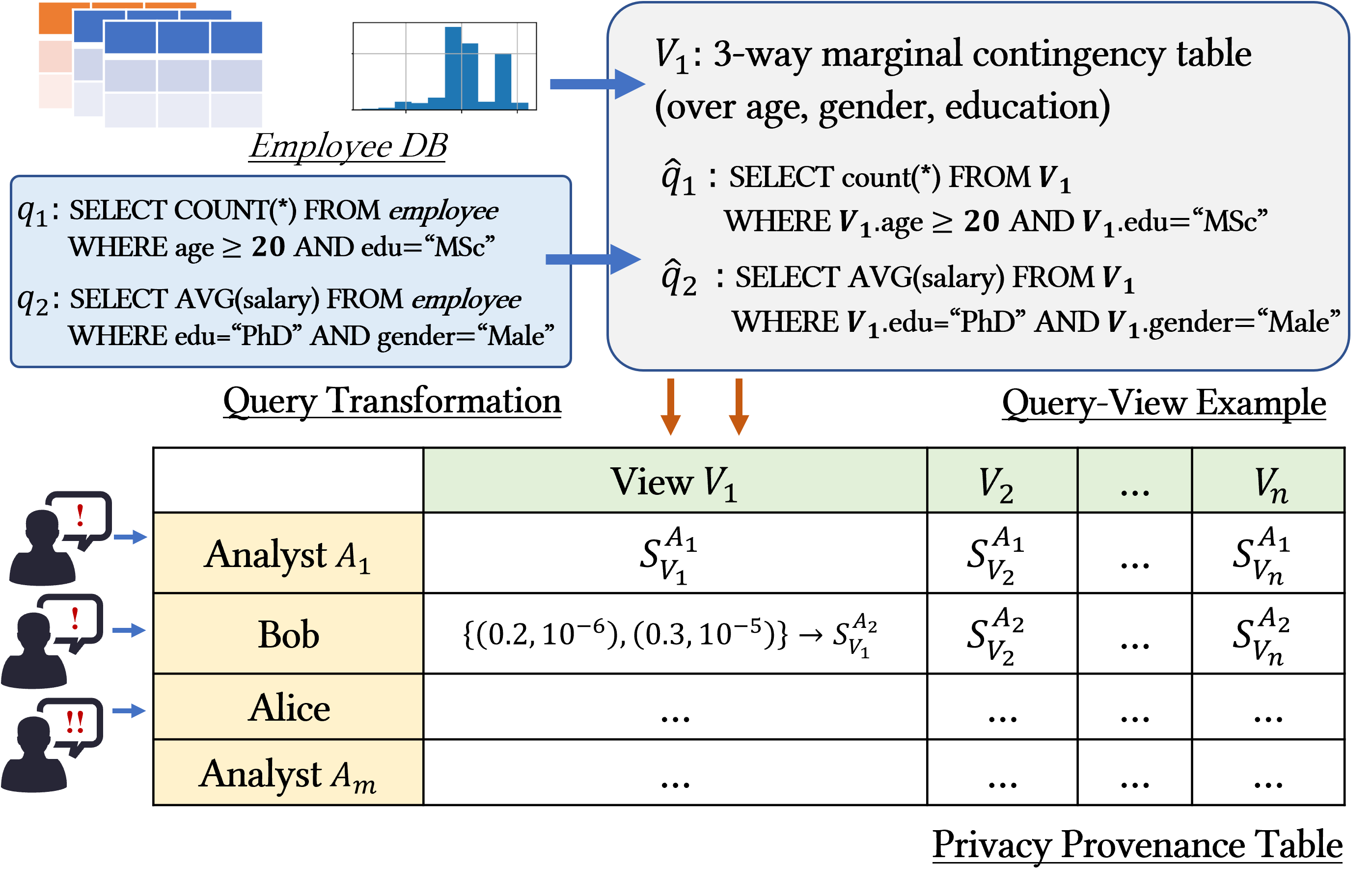}
  
    \caption{Illustration of the Histogram View, the Query Transformation and the Privacy Provenance Table: 
    \circledtext{1} Histogram view $V_1$ over age, gender, and education is built on the database snapshot; 
    \circledtext{2} Forthcoming queries $q_1$ and $q_2$ are transformed into linearly answerable queries $\hat{q}_1$ and $\hat{q}_2$ over $V_1$; 
    \circledtext{3} Analysts ($A_1, A_2$ with low, $A_3$ with high privilege) are recorded in the provenance table -- the privacy loss for each view to each analyst is tracked for real-time queries.
    }
    \label{fig:provenance_table}
   
\end{figure}

\subsection{Privacy Provenance Table}
\label{sec:provT}

To meet the first two principles, we propose a privacy provenance table for \oursystem, inspired by the access matrix model in access control literature \cite{samarati2000access}, to track the privacy loss per analyst and per view, and further bound the privacy loss. 
Particularly, in our model, the state of the overall privacy loss of the system is defined as a triplet $(\mc{A}, \mc{V}, \mc{P})$, where $\mc{A}$ denotes the set of data analysts and $\mc{V}$ represents the list of query-views maintained by the system. 
We denote by $\mc{P}$ the privacy provenance table,  defined as follows.

\begin{definition}[Privacy Provenance Table] 
The privacy provenance table $\mc{P}$ consists of (i) a provenance matrix $P$ that tracks the privacy loss of a view in $\mc{V}$ to each data analyst in $\mc{A}$, where each entry of the matrix $P[A_i, V_j]$
records the current cumulative privacy loss $S^{A_i}_{V_j}$, on view $V_j$ to analyst $A_i$;
(ii) a set of row/column/table constraints, $\Psi$:
a \emph{row constraint} for $i$-th row of $P$, denoted by $\psi_{A_i}$, refers to the allowed maximum privacy loss to a data analyst $A_i\in \mc{A}$ (according to his/her privilege level); a \emph{column constraint} for the $j$-th column, denoted by $\psi_{V_j}$ refers to as the allowed maximum privacy loss to a specific view $V_j$; the \emph{table constraint} over $P$, denoted by $\psi_P$, specifies the overall privacy loss allowed for the protected database.
\end{definition}

The privacy constraints and the provenance matrix are correlated. In particular, the row/column constraints cannot exceed the overall table constraint, and each entry of the matrix cannot exceed row/column constraints. The correlations, such as the composition of the privacy constraints of all views or all analysts, depend on the DP mechanisms supported by the system. We provide the details of DP mechanisms and the respective correlations in \provT in Section~\ref{sec:low_sensitive_algos}.

\begin{example}
\label{example:provT}
Figure~\ref{fig:provenance_table} gives an example of the \provT for $n$ views and $m$ data analysts. When \oursystem receives query $q_1$ from Bob, it plans to use view $V_1$ to answer it. \oursystem first retrieves the previous cumulative cost of $V_1$ to Bob from the matrix, $P[Bob, V_1]$, and then computes the new cumulative cost $S_{V_1}^{Bob}$ for $V_1$ to Bob as if it answers $q_1$ using $V_1$. If the new cost $S_{V_1}^{Bob}$ is smaller than Bob's privacy constraint $\psi_{Bob}$, the view constraint $\psi_{V_1}$, and the table constraint $\psi_P$, \oursystem will answer $q_1$  and update $P[Bob, V_1]$ to $S_{V_1}^{Bob}$; otherwise, $q_1$ will be rejected. 
\end{example}

Due to the privacy constraints imposed by the \provT, queries can be rejected when the cumulative privacy cost exceeds the constraints. \oursystem needs to design DP mechanisms that well utilize the privacy budget to answer more queries. Hence, we formulate the \emph{maximum query answering problem} based on the \provT.

\begin{problem}
Given a \provT $(\mc{A}, \mc{V}, \mc{P})$,  at each time, a data analyst $A_i\in \mc{A}$ submits the query with a utility requirement $(q_i, v_i)$, where the transformed $\hat{q}_i\in \mc{V}$, how can we design a system to answer as many queries as possible without violating the row/column/table privacy constraints in $P$ while meeting the utility requirement per query?
\end{problem}

\ifarxiv
Hence we would like to build an efficient system to satisfy the aforementioned design principles and enable the \provT, and develop algorithms to solve the maximum query answering problem in \oursystem.
We next outline the system modules in \oursystem, and then provide detailed algorithm design in Section~\ref{sec:low_sensitive_algos}. 
\fi

\subsection{System Modules}
\label{sec:overview_modules}

The \oursystem system works as a middle-ware between data analysts and existing DP DBMS systems (such as PINQ, Chorus, and PrivateSQL) to provide intriguing and add-on functionalities, including fine-grained privacy tracking, view/synopsis management, and privacy-accuracy translation.
We briefly summarize the high-level ideas of the modules below.

\stitle{Privacy Provenance Tracking.}
\oursystem maintains the \provT for each registered analyst and each generated view, as introduced in Section \ref{sec:provT}.
Constraint checking is enabled based on this provenance tracking to decide whether to reject 
an analyst's query or not.
We further build DP mechanisms to maintain and update the DP synopses and the \provT.

\stitle{Dual Query Submission Mode.}
\oursystem provides two query submission modes to data analysts.
\textit{Privacy-oriented mode} \cite{mcsherry2009privacy,johnson2018towards,johnson2020chorus,zhang2018ektelo}: queries are submitted with a pre-apportioned privacy budget, i.e., $(A_i, q_i, \{\epsilon_i, \delta_i \})$.
\textit{Accuracy-oriented mode} \cite{ge2019apex,lobo2020programming}:
analysts can write the query with a desired accuracy bound, i.e., queries are in form of $(A_i, q_i, v_i)$. We illustrate our algorithm with the accuracy-oriented mode.

\begin{algorithm}[t]
    \SetAlgoLined
    \DontPrintSemicolon
    \KwInput{Analysts $\mathcal{A} \coloneqq A_1, \dots, A_n$; Database instance $D$; 
    Privacy provenance table $\mc{P}=(P,\Psi)$.}

    \Admin sets up the \provT 
    $\mc{P}$ \\
    
    Initialize all the synopses for $V\in \mc{V}$\\

    \revise{Set $\delta$ in the system} \\
    
    \Repeat{No more queries sent by analysts}{
        Receive $(q_i, v_i)$ from data analyst $A_i$ \\
        Select view $V\in \mc{V}$ to answer query $q_i$ \\
        Select mechanism $M \in \mc{M}$ applicable to $q_i$ \\
        
        $\epsilon_i \gets$ \textsc{M.privacyTranslate}($q_i, v_i, V, p$) \label{alg:overview_privacy_translator} \\

        \uIf{\textsc{\emph{M.constraintCheck(}$P, A_i, V, \epsilon_i, \Psi$\emph{)}}}{
            
            $V_{A_i}^{\epsilon'} \gets$\textsc{M.run($P, A_i, V, \epsilon_i$)} \\
            Answer $q_i$ with $V_{A_i}^{\epsilon'}$ and return answer $r_i$ to $A_i$
        }
        \Else{
            reject the query $q_i$
        }
        
    }
    \caption{System Overview}
    \label{alg:overview}
\end{algorithm}

\stitle{Algorithm Overview.}
Algorithm~\ref{alg:overview} summarizes how \oursystem uses the DP synopses to answer incoming queries. 
At the system setup phase (line 1-3), the \admin initializes the \provT by setting the matrix entry as 0 and the row/column/table constraints $\Psi$.
The system initializes empty %
synopses for each view.  The data analyst specifies a query $q_i$ with its desired utility requirement $v_i$ (line 5). 
Once the system receives the request, it selects the suitable view and mechanism to answer this query (line 6-7)
and uses the function \textsc{privacyTranslate()} to find the minimum privacy budget $\epsilon_i$ for $V$ to meet the utility requirement of $q_i$ 
(line 8). 
Then, \oursystem checks if answering $q_i$ with budget $\epsilon_i$ will violate the privacy constraints $\Psi$ (Line 9).
If this constraint check passes, we run the mechanism to obtain a noisy synopsis (line 10). 
\oursystem uses this synopsis to answer query $q_i$ and returns the answer to the data analyst (line 11). If the constraint check fails, \oursystem rejects the query (line 13). We show concrete DP mechanisms with their corresponding interfaces in the next section. 

\revise{
\stitle{Remark.}\linelabel{line:delta_1}
For simplicity, we drop $\delta$ and focus on $\epsilon$ as privacy loss/budget in privacy composition, but $\delta$'s of DP synopses are similarly composited as stated in Theorem~\ref{thm:sequential_composition_multi_dp}. For the accuracy-privacy translation, we consider a fixed given small delta and aim to find the smallest possible epsilon to achieve the accuracy specification of the data analyst.
}

\section{DP Algorithm Design}
\label{sec:low_sensitive_algos}

In this section, we first describe a \baseline DP mechanism that can  instantiate the system interface but cannot maximize the number of queries being answered.
Then we propose an additive Gaussian mechanism that leverages the correlated noise in query answering to improve the %
utility of the \baseline mechanism.
Without loss of generality, we assume the data analysts do not submit the same query with decreased accuracy requirement (as they would be only interested in a more accurate answer).

\subsection{\Baseline Approach}
\label{sec:vanilla}

\begin{algorithm}[tbp]
    \SetAlgoLined
    \DontPrintSemicolon
    
    Set $\delta$ in the system \\
    \SetKwProg{Fn}{Function}{ :}{end}
    \Fn{\textsc{\emph{run(}$P, A_i, V, \epsilon_i$}\emph{)} \label{func:baseline_run}}{
        Generate a synopsis $V_{A_i}^{\epsilon_i}$ from view $V$
                \\
        Update \provT $P[A_i,V]\gets P[A_i,V]+\epsilon_i$ \\
        \textbf{return} $r_i \gets V_{A_i}^{\epsilon_i}$
        \label{func:baseline_run_end}
    }
    \Fn{\textsc{\emph{privacyTranslate(}}$q_i, v_i, V, p$\emph{)} \label{fuc:baseline_translation}}{
        Set $u = \psi_P,~ l = 0$ \\
        $v \gets$ \textsc{calculateVariance}($q_i, v_i, V$) \label{baseline:calculatevariance} \\
        $\epsilon = $ \textsc{binarySearch}($l, u, $\textsc{testAccuracy}($\cdot, v, \Delta q_i, \delta$), \revise{$p$}) \label{baseline:binary} \\
        \textbf{return} $\epsilon$
        \label{fuc:baseline_translation_end}
    }
    \Fn{\textsc{\emph{constraintCheck(}}$P, A_i, V_j, \epsilon_i, \Psi$\emph{)} \label{func:baseline_constraint} }{

         \uIf{$(P$.\textsc{\emph{composite()}} + $ \epsilon_i \leq \Psi.\psi_{P}) \land$ $(P$.\textsc{\emph{composite(}}axis=Row\emph{)} + $ \epsilon_i \leq \Psi.\psi_{A_i}) \land$ $(P$.\textsc{\emph{composite(}}axis=Column\emph{)} + $ \epsilon_i \leq \Psi.\psi_{V_j})$}{
            \textbf{return} True/Pass
        }
        \label{func:baseline_constraint_end}
    }   
    \caption{\Baseline Approach}
    \label{alg:baseline}
\end{algorithm}

The \baseline approach is based on 
the Gaussian mechanism (applied to both the basic Gaussian mechanism \cite{dwork2014algorithmic} and the analytic Gaussian mechanism \cite{BW18analytic}).
We describe how the system modules %
are instantiated with the \baseline approach.

\subsubsection{Accuracy-Privacy Translation}
This module interprets the user-specified utility requirement into the minimum privacy budget \revise{\linelabel{line:translation_1}
(Algorithm \ref{alg:baseline}: line \ref{fuc:baseline_translation}-\ref{fuc:baseline_translation_end}).
Note that instead of perturbing the query result to $q_i$ directly, we generate a noisy DP synopsis and use it to answer a query (usually involving adding up a number of noisy counts from the synopsis). 
Hence, we need to translate the accuracy bound $v_i$ specified by the data analyst over the query to the corresponding accuracy bound $v$ for the synopsis before searching for the minimal privacy budget (Algorithm \ref{alg:baseline}: line \ref{baseline:calculatevariance}). 
Here, $v$ represents the variance of the noise added to each count of the histogram synopsis.
Next, we search for the minimal privacy budget $\epsilon$ that results in a noisy synopsis with  noise variance not more than $v$, based on the following analytic Gaussian translation.
}

\begin{definition}[Analytic Gaussian Translation]
\label{def:analytic_gaussian_translation}
Given a %
query $q: \mathcal{D} \rightarrow \mathbb{R}^d$ %
to achieve an expected squared error bound $v$ for this query, the minimum privacy budget for the analytic Gaussian mechanism should satisfy
$
\Phi_{\mc{N}} \left( \frac{\Delta q}{2 v} - \frac{\epsilon v}{\Delta q} \right)
    - e^{\epsilon} \Phi_{\mc{N}} \left( - \frac{\Delta q}{2 v} - \frac{\epsilon v}{\Delta q} \right) 
    \leq \delta.
$
That is, given $\Delta q, \delta, v$, to solve the following optimization problem to find the minimal $\epsilon$.
\begin{eqnarray}
    \min_{\epsilon\in (0,\psi_P]} \epsilon  
    \mbox{ s.t. }  \Phi_{\mc{N}} \left( \frac{\Delta q}{2 v} - \frac{\epsilon v}{\Delta q} \right)
    - e^{\epsilon} \Phi_{\mc{N}} \left( - \frac{\Delta q}{2 v} - \frac{\epsilon v}{\Delta q} \right) 
    \leq \delta \label{AGT_constraint} %
\end{eqnarray}
\end{definition}
Finding a closed-form solution to the problem above is not easy. However, we observe that the LHS of the constraint in Equation \eqref{AGT_constraint} is a monotonic function of $\epsilon$. 
\revise{\linelabel{line:translation_2}
Thus, we use binary search (Algorithm \ref{alg:baseline}: line \ref{baseline:binary}) to look for the smallest possible value for $\epsilon$. 
For each tested value of $\epsilon$, we compute its analytic Gaussian variance~\cite{BW18analytic}, denoted by $v'$. If $v'>v$, then this value is invalid and we search for a bigger epsilon value; otherwise, we search for a smaller one. We stop at an epsilon value with a variance $v'\leq v$ and a distance $p$ from the last tested invalid epsilon value. We have the following guarantee for this output. 
}

\revise{\linelabel{line:translation_3}
\begin{proposition}[Correctness of Translation]\label{thm:translation_correctness}
Given a query $(q_i, v_i)$ 
and the view $V$ for answering $q_i$, the translation function {(Algorithm~\ref{alg:baseline}, \textsc{privacyTranslation})} outputs a privacy budget $\epsilon$.
The query $q_i$ can then be answered with expected square error $v_q$ 
over the updated synopsis $V^{\epsilon}_{A}$ such that: i) meets the accuracy requirement $v_q \leq v_i$, and ii) $\epsilon - \epsilon^* \leq p$, where $\epsilon^*$ is the minimal privacy budget to meet the accuracy requirement for Algorithm~\ref{alg:baseline} (\textsc{run} function).
\end{proposition}

\ifarxiv
\begin{proof}[Proof Sketch]
First, line~\ref{baseline:calculatevariance} derives the per-bin accuracy requirement based on the submitted per-query accuracy, and plugs it into the search condition (Equation~\eqref{AGT_constraint}).
Note that our DP mechanism and the accuracy requirement are data-independent.
As long as the searching condition holds, the noise added to the query answer in \textsc{run} satisfies $v_q \leq v_i$.
Second, the stopping condition of the search algorithm guarantees 1) there is a solution, 2) the searching range is reduced to $\leq p$. Thus we have $\epsilon - \epsilon^* \leq p$. 
\end{proof}
\fi
}

\subsubsection{Provenance Constraint Checking}
\label{subsubsec:prov_checking}

As mentioned, the \admin can specify privacy constraints in \provT. \oursystem decide whether \emph{reject or answer a query} using the provenance matrix $P$ and the privacy constraints $\Psi$ in \provT, as indicated in Algorithm \ref{alg:baseline}: \ref{func:baseline_constraint}-\ref{func:baseline_constraint_end} (the function \textsc{constraintCheck}).   
This function checks whether the three types of constraints would be violated when the current query was to issue.
The \textsc{composite} function in this constraint-checking algorithm can refer to the basic sequential composition or tighter privacy composition given by Renyi-DP \cite{mironov2017renyi} or zCDP \cite{dwork2016concentrated,bun2016concentrated}.
We suggest to use advanced composition for accounting privacy loss over time, but not for checking constraints, because the size of the provenance table $n * m$ is too small for a tight bound by this composition.

\subsubsection{Putting Components All Together}
The \baseline approach is aligned with existing DP query systems in the sense that it adds independent noise to the result of each query. Hence, it can be quickly integrated into these systems to provide privacy provenance and accuracy-aware features with little overhead.
Algorithm \ref{alg:baseline}: \ref{func:baseline_run}-\ref{func:baseline_run_end} (the function \textsc{run}) outlines how the \baseline method runs. 
It first generates the DP synopsis $V^{\epsilon_i}_{A_i}$ using analytic Gaussian mechanism for the chosen view $V$ and updates the corresponding entry $P[A_i, V]$ in the \provT by compositing the consumed privacy loss $(\epsilon_i, \delta_i)$ on the query (depending on the specific composition method used).
We defer the analysis for the accuracy and privacy properties of the \baseline mechanism to Section~\ref{sec:guarantees}.

\subsection{Additive Gaussian Approach}
\label{subsec:additive_gaussian_mechanism}
\shepherding{\linelabel{re-position}
While ideas of using correlated Gaussian noise have been exploited \cite{LiCLZ12multilevel}, we adopt similar statistical properties into an additive Gaussian DP mechanism, a primitive to build our additive Gaussian approach for synopses maintenance. Then, we describe how \oursystem generates and updates the (local and global) DP synopses with this algorithm across analysts and over time.}

\subsubsection{Additive Gaussian Mechanism} \label{sec:aGM}
The additive Gaussian mechanism (additive GM or aGM) modifies the standard Gaussian mechanism, based on the nice statistical property of the Gaussian distribution ---  the sum of i.i.d. normal random variables is still normally distributed. We outline this primitive mechanism in Algorithm \ref{alg:additive_gaussian_noise_adding}.
This primitive takes a query $q$, a database instance $D$, a set of privacy budgets $\mc{B}$ corresponding to the set of data analysts $\mc{A}$ as input, and this primitive outputs a noisy query result to each data analyst $A_i$, which consumes the corresponding privacy budget $(\epsilon_i,\delta)$.
Its key idea is only to execute the query (to get the true answer on the database) once, and cumulatively inject noises to previous noisy answers, when multiple data analysts ask the same query.
In particular, we sort the privacy budget set specified by the analysts.
Starting from the largest budget, we add noise w.r.t the Gaussian variance $\sigma_i^2$ calculated from the query sensitivity $\Delta q$ and this budget $(\epsilon_i, \delta)$.
For the rest of the budgets in the set, we calculate the Gaussian variance $\sigma_j^2$ in the same approach but add noise w.r.t $\sigma_j^2 - \sigma_i^2$ to the previous noisy answer.
The algorithm then returns the noisy query answer to each data analyst.
The privacy guarantee of this primitive is stated as follows.

\begin{algorithm}[t]
    \SetAlgoLined
    \DontPrintSemicolon
    \KwInput{Analysts $\mathcal{A} \coloneqq A_1, \dots, A_n$; A query $q$; Database instance $D$; A set of privacy budgets $\mc{B} \revise{ \coloneqq \{ \mc{E}, \mc{D} \} =} (\epsilon_1, \delta), (\epsilon_2, \delta), \dots, (\epsilon_n, \delta)$.}
    \KwOutput{A set of noisy answers $r_1, r_2, \dots, r_n$.}
    \SetKwProg{Fn}{Function}{ :}{end}
    \Fn{\textsc{\emph{additiveGM(}$\mathcal{A}, \mc{B}, q, D$}\emph{)}}{
        $r \gets$ \textsc{queryExec}($q, D$) \Comment{Obtain true query answer.} \\
        $\Delta q \gets$ \textsc{sensCalc}($q$) \Comment{Sensitivity calculation.} \\
        $\mc{B}' \gets$ \textsc{sort}($\mc{B}, \epsilon_i$) \Comment{Sort $\mc{B}$ on the desc order of $\epsilon$'s.} \label{alg_line:agm_sort} \\
        $(\epsilon_i, \delta) \gets$ \textsc{pop}($\mc{B}'$) \Comment{Pop the 1st element.} \\
        $\sigma_i \gets $ \textsc{analyticGM}($\epsilon_i, \delta, \Delta q $) 
        \Comment{\cite{BW18analytic}} \label{alg_line:agm_sigma} \\
        $r_i \gets r + \eta_i \sim \mathcal{N}(0, \sigma_i^2)$ \Comment{Add Gaussian noise.} \\
        \While{$\mc{B}' \neq \emptyset $}{
            $(\epsilon_j, \delta) \gets$ \textsc{pop}($\mc{B}'$) \\
            $\sigma_j \gets $ \textsc{analyticGM}($\epsilon_j, \delta, \Delta q $) %
             \Comment{\cite{BW18analytic}} \\
            $r_j \gets r_i + \eta_j \sim \mathcal{N}(0, \sigma_j^2 - \sigma_i^2)$;
        }
        \textbf{return} $\mc{R} \coloneqq \{ r_i | i \in [n] \}$;
    }
    \caption{Additive Gaussian Noise Calibration}
    \label{alg:additive_gaussian_noise_adding}
\end{algorithm}

\begin{theorem}
\label{thm:additiveGM_privacy}
Given a database $D$, a set of privacy budgets $\mc{B} \coloneqq (\epsilon_1, \delta), (\epsilon_2, \delta), \dots, (\epsilon_n, \delta)$ and a query $q$, the additive Gaussian mechanism (Algorithm \ref{alg:additive_gaussian}) that returns a set of noisy answers $r_1, r_2, \dots, r_n$ to each data analyst $A_i$ %
satisfies $[(A_1, \epsilon_1, \delta), ...,$ $(A_n, \epsilon_n, \delta)]$-multi-analyst-DP and $(\max \{ \epsilon_1, \dots, \epsilon_n \}, \delta)$-DP.
\end{theorem}

\begin{proof}[Proof Sketch]
To each data analyst $A_i$, the DP mechanism is equivalent to the standard Gaussian mechanism with a proper variance that satisfies $(\epsilon_i,\delta)$-DP. 
Since the data is looked at once, the $(\max \{ \epsilon_1, \dots, \epsilon_n \}, \delta)$-DP is guaranteed by post-processing.
\end{proof}

\ifarxiv
\revise{
\stitle{Discussion on $\delta$.} \linelabel{line:delta_2}
If the $\delta$ is not a fixed parameter in the system, it could happen in privacy budget $(\epsilon_i, \delta_i)$ that $\epsilon_i = \max \mc{E}$ but $\delta_i = \min \mc{D}$.
Algorithm~\ref{alg:additive_gaussian_noise_adding} can be simply modified to handle this by not sorting $\mc{B}$ based on descending $\epsilon$'s (line \ref{alg_line:agm_sort}) but according to the ascending order of the calculated $\sigma$ (line \ref{alg_line:agm_sigma}).
}
\fi

\subsubsection{Synopses Management}
\label{subsubsec:synopses_updating}

We introduce the concept of global and local DP synopses and then discuss the updating process in our additive GM.
A DP synopsis (or synopsis for short) is a noisy answer to a (histogram) view over a database instance. We first use the privacy-oriented mode to explain the synopses management for clarity, and then elaborate the accuracy-oriented mode in the accuracy-privacy translation module (Section~\ref{subsubsec:agm_translation}).

\stitle{Global and Local DP Synopses.} To solve the maximum query answering problem, 
for each view $V\in \mc{V}$, 
\oursystem maintains a \emph{global DP synopsis} with a cost of $(\epsilon,\delta)$, denoted by $V^{\epsilon,\delta}(D)$ or $V^{\epsilon}$, where $D$ is the database instance. For simplicity, we drop 
$\delta$ by considering the same value for all $\delta$ and $D$. For this veiw,  \oursystem also maintains a \emph{local DP synopsis} for each analyst $A_i\in \mc{A}$, denoted by $V_{A_i}^{\epsilon'}$, where the local synopsis is always generated from the global synopsis $V^{\epsilon}$ of the view $V$ by adding more noise. Hence, we would like to ensure $\epsilon\geq\epsilon'$. This local DP synopsis  $V_{A_i}^{\epsilon'}$ will be used to answer the queries asked by the data analyst $A_i$.

The process of updating synopses consists of two parts.
The first part is to update the local synopses based on the global synopses.
The second part is to update the global synopses by relaxing the privacy guarantee, in order to answer a query with a higher accuracy requirement.
We discuss the details below.

\stitle{Generating Local Synopses from Global Synopses.}
We leverage our %
additive GM 
primitive to release a local DP synopsis $V_{A_i}^{\epsilon'}$ from a given global synopsis $V^{\epsilon}$, where $V^{\epsilon}$ is generated by a Gaussian mechanism.
Given the privacy guarantee $\epsilon$ (and $\delta$) and the sensitivity of the view, the Gaussian mechanism can calculate a proper variance $\sigma^2$ for adding noise and ensuring DP.
The additive GM calculates $\sigma^2$ and $\sigma'^2$ based on $\epsilon$ and $\epsilon'$ respectively, and then generates the local synopsis $V_{A_i}^{\epsilon'}$ by injecting independent noise drawn from $\mc{N}(0, \sigma^{\prime 2} - \sigma^2)$ to the global synopsis $V^{\epsilon}$. As the global synopsis is hidden from all the analysts, the privacy loss to the analyst $A_i$ is $\epsilon'$. Even if all the analysts collude, the maximum privacy loss is bounded by the budget spent on the global synopsis. 

\begin{figure*}
    \centering
    \includegraphics[width=0.9\textwidth]{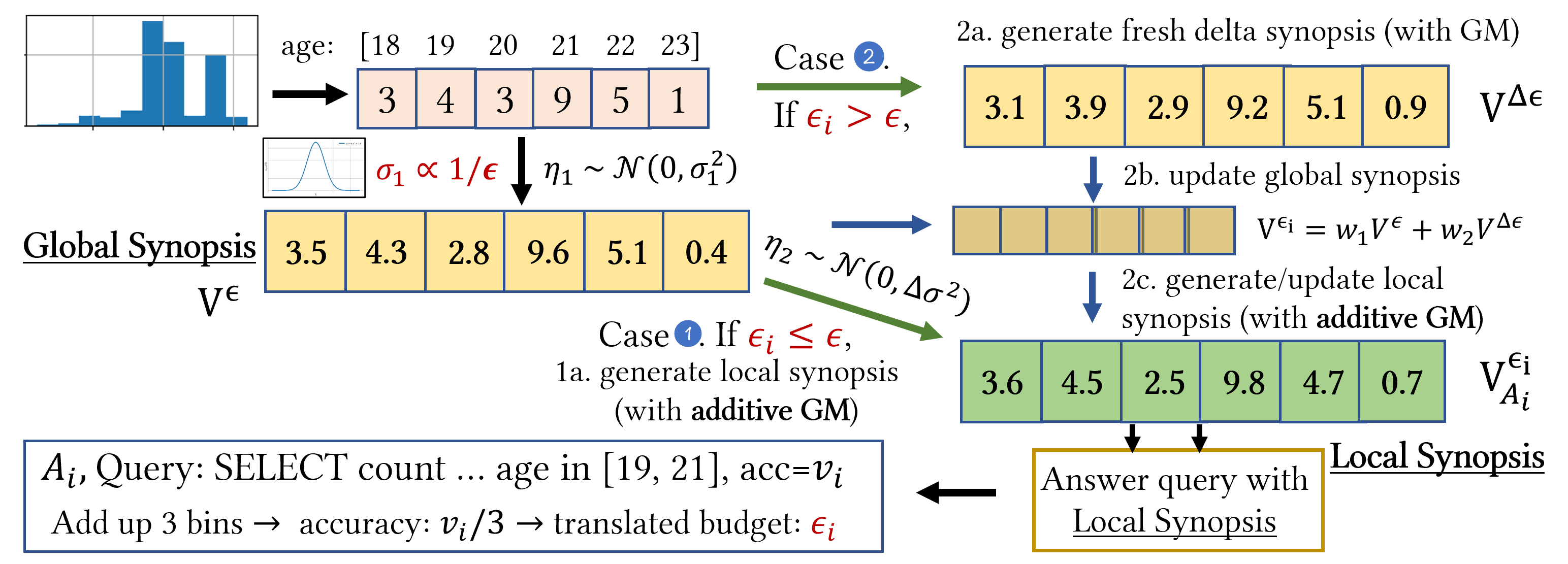}
    \caption{Illustration of the Additive Gaussian Approach}
    \label{fig:agm}
\end{figure*}

\begin{example}
    Alice and Bob are asking queries to \oursystem.
    Alice asks the first query $q_1$ (which is answerable on $V_1$) with \emph{budget requirement} $\epsilon_{V_1, Alice} = 0.5$.
    \oursystem generates a global synopsis $V^{0.5}$ for $V$ with budget 0.5 and then generate a local synopsis $V^{0.5}_{Alice}$ from the global synopsis $V^{0.5}$ for Alice.
    Bob next asks query $q_2$ (which is also answerable on $V_1$) with budget $\epsilon_{V_1, Bob} = 0.3$. Since the budget $0.3 < 0.5$, we use the additive GM to generate a local synopsis $V^{0.3}_{Bob}$ from the global synopsis $V^{0.5}$ for Bob and return the query answer based on the local synopsis $V^{0.3}_{Bob}$.
    This example follows ``Case \circledtext{1}'' in Fig. \ref{fig:agm}.
\end{example}

\stitle{Updating Global Synopses by Combining Views.}
When the global DP synopsis $V^{\epsilon}$ is not sufficiently accurate to handle a local synopsis request at privacy budget $\epsilon_t$, \oursystem spends additional privacy budget $\Delta\epsilon$ to update the global DP synopsis to 
$V^{\epsilon+\Delta\epsilon}$, where $\Delta\epsilon=\epsilon_t -\epsilon$. 
We still consider Gaussian mechanism, which generates an intermediate DP synopsis $V^{\Delta\epsilon}$ with a budget $\Delta\epsilon$.
Then we combine the previous synopses with this intermediate synopsis into an updated one.
The key insight of the combination is to properly involve the fresh noisy synopses by assigning each synopsis
with a weight proportional to the inverse of its noise variance, which gives the smallest expected square error based on UMVUE \cite{kiefer1952minimum,rao1949sufficient}.
That is, for the $t$-th release, we combine these two synopses:  
\begin{equation}
\label{equ:view_combine}
    V^{\epsilon_t} = (1-w_t) V^{\epsilon_{t-1}} + w_t V^{\Delta\epsilon}.
\end{equation}
The resulted expected square error for $V^{\epsilon_t}$ is $v_t=(1-w_t)^2 v_{t-1} + w_t^2 v_\Delta$, where $v_{t-1}$ is the noise variance of view $V^{\epsilon_{t-1}}$, and
$v_\Delta$ is derived from $V^{\Delta\epsilon}$.
To minimize the resulting error, $w_t=\frac{v_{t-1}}{v_\Delta + v_{t-1}}$.

\begin{example}
    At the next time stamp, Bob asks a query $q_1$ with budget $\epsilon_{V_1, Bob} = 0.7$.
    Clearly the current global synopsis $V^{0.5}$ is not sufficient to answer this query because $0.7 > 0.5$.
    Then the system needs to update $V^{0.5}$ and this is done by: 1) first generating a fresh global synopsis $V^{0.2}$ using analytic GM from $V(D)$;
    2) then combining it with $V^{0.5}$ to form $V^{0.7}$ using Equation~\eqref{equ:view_combine}.
    This example is illustrated as steps 2a and 2b of ``Case \circledtext{2}'' in Fig. \ref{fig:agm}.
\end{example}

\begin{lemma}[Correctness of View Combination]
Releasing the combined DP synopsis in the $t$-th update is $(\epsilon_{t-1} + \Delta\epsilon, \delta_{t-1} + \Delta\delta)$-DP.
\end{lemma}

\ifarxiv
\begin{proof}[Proof Sketch]
    The $t$-th update combines an $(\epsilon_{t-1}, \delta_{t-1})$-DP synopsis and a fresh synopsis that is $(\Delta\epsilon, \Delta\delta)$-DP.
    By sequential composition, the combined synopsis is $(\epsilon+\Delta\epsilon, \delta + \Delta\delta)$-DP.
\end{proof}
\fi

The view combination is not \emph{frictionless}.
Although the combined synopsis $V^{\epsilon+\Delta\epsilon}$ achieves $(\epsilon+\Delta\epsilon, \delta + \Delta\delta)$-DP, if we spend the whole privacy budget on generating a synopsis all at once, this one-time synopsis $V^*$ has the same privacy guarantee but is with less expected error than $V^{\epsilon+\Delta\epsilon}$.
We can show by sequentially combining and releasing synopses over time it is optimal among all possible linear combinations of synopses,
however, designing a frictionless updating algorithm for Gaussian  mechanisms is non-trivial in its own right, which remains for our future work.

\begin{theorem}[Optimality of Linear View Combination]
Given a sequences of views, $V^{\epsilon_1}, V^{\Delta \epsilon_2}, \ldots, V^{\Delta \epsilon_t}$, the expected squared error of our $t$-th release is less than or equal to that of releasing $w_1' V^{\epsilon_1} + \sum_{i=2}^t w_i'V^{\Delta \epsilon_i}$ for all $\{ w_i \mid i = 1, \ldots, t\}$ s.t. $\sum_i w_i' = 1$.  
\end{theorem}
Intuitively, this theorem is proved by a reduction to the best unbiased estimator and re-normalizing weights based on induction.

\stitle{Updating Local Synopses and Accounting Privacy.}
When a local DP synopsis is not sufficiently accurate to handle a query, but the budget request for this query $\epsilon_i$ is still smaller than or equal to the budget $\epsilon_t$ for the global synopsis of $V$, \oursystem generates a new local synopsis $V^{\epsilon_i}_{A_i}$ from $V^{\epsilon_t}$ using additive GM.
The analyst $A_i$ is able to combine query answers for a more accurate one, but the privacy cost for this analyst on view $V$ is bounded as $\min (\epsilon_t, P[A_i, V]+ \epsilon_i)$ which will be updated to $P[A_i, V]$.

\begin{example}
    Analysts' queries are always answered on local synopses.
    To answer Bob's query $(q_1, \epsilon_{V_1, Bob} = 0.7)$, \oursystem uses additive GM to generate a fresh local synopsis $V^{0.7}_{Bob}$ from $V^{0.7}$ and return the answer to Bob.
    Alice asks another query $(q_1, \epsilon_{V_1, Alice} = 0.6)$. \oursystem updates $V^{0.5}_{Alice}$ by generating $V^{0.6}_{Alice}$ from $V^{0.7}$.
    Both analysts' privacy loss on $V$ will be accounted as $0.7$.
    This example complements the illustration of step 2c of ``Case \circledtext{2}'' in Fig. \ref{fig:agm}.
\end{example}

\subsubsection{Accuracy-Privacy Translation}
\label{subsubsec:agm_translation}

The accuracy translation algorithm should consider the friction at the combination of global synopses.
We propose an accuracy-privacy translation paradigm (Algorithm \ref{alg:additive_gaussian}: \ref{func:aGM_translation}) involving this consideration. 
This translator module takes the query $q_i$, the utility requirement $v_i$, the synopsis $V$ 
for answering the query, and additionally the current global synopsis $V^{\epsilon}$ (we simplify the interface in Algorithm \ref{alg:overview}) as input, and outputs the corresponding budget \revise{$\epsilon_i$ for \textsc{run}} (omitting the same value $\delta$).

As the first query release for the view does not involve the frictional updating issue, we separate the translation into two phases where the first query release directly follows the analytic Gaussian translation in our \baseline approach.
For the second phase, given a global DP synopsis $V^{\epsilon_c}$ at hand (with \emph{tracked} expected error $v'$)
for a specific query-view and a new query is submitted by a data analyst with expected error $\revise{v_i} < v'$, we solve an optimization problem to find the Gaussian variance of the fresh new DP synopsis. 
We first calculate the Gaussian variance of the current DP synopsis $v'$ (line \ref{line:aGM_translation_estimate}) and solve the following optimization problem (line \ref{line:aGM_translation_optimize}).
\begin{eqnarray}
    \argmax_{\sigma} ~~  \revise{v_i} = w^2 v' + (1-w)^2 v_t ~~ 
    \mbox{ s.t. }  
    w \in [0, 1]
\end{eqnarray}

The solution gives us the minimal error variance $v_t=\sigma^2$ (line \ref{line:aGM_translation_vanilla}). 
By translating $\sigma^2$ into the privacy budget using the standard analytic Gaussian translation technique (Def. \ref{def:analytic_gaussian_translation}), we can get the minimum privacy budget that achieves the required accuracy guarantee (line \ref{func:aGM_translation_end}).
Note that when the requested accuracy $\revise{v_i} > v'$, the solution to this optimization problem is $w=0$, which automatically degrades to the \baseline translation.

\begin{theorem}\label{thm:agm_translation_correctness}
    Given a query $(q_i, v_i)$ 
and the view $V$ for answering $q_i$, the translation function (Algorithm~\ref{alg:additive_gaussian}, \textsc{privacyTranslation}) outputs a privacy budget $\epsilon$.
The query $q_i$ can then be answered with expected square error $v_q$ 
over the updated synopsis $V^{\epsilon}_{A}$ such that: i) meets the accuracy requirement $v_q \leq v_i$, and ii) $\epsilon - \epsilon^* \leq p$, where $\epsilon^*$ is the minimal privacy budget to meet the accuracy requirement for Algorithm~\ref{alg:additive_gaussian} (\textsc{run} function).
\end{theorem}

\ifarxiv
\begin{proof}[Proof Sketch]
    The \textsc{privacyTranslation} in additive Gaussian approach calls the translation function in the vanilla approach as subroutine (Algorithm~\ref{alg:additive_gaussian}: \ref{line:aGM_translation_vanilla}).
    The correctness of the additive Gaussian \textsc{privacyTranslation} depends on inputting the correct expected square error, which is calculated based on the accuracy requirement $v_i$ while considering the frictions, into the subroutine. The calculation of the expected squared error with an optimization solver has been discussed and analyzed above.
\end{proof}
\fi

\begin{algorithm}[t]
    \SetAlgoLined
    \DontPrintSemicolon
    \SetKwProg{Fn}{Function}{ :}{end}
    Set $\delta$ in the system \\
    \Fn{\textsc{\emph{run(}$q, \mc{P}, A_i, V_{A_i}, \epsilon_i$}\emph{)} \label{func:aGM_run}}{
        \uIf{$\epsilon_i>\epsilon$ for global synopsis $V^{\epsilon}$}{
                $\Delta\epsilon = \epsilon_i - \epsilon$ \\
                 $V^{\Delta\epsilon} \gets V + \eta \sim \mc{N}(0, $  \textsc{analyticGM}($\Delta\epsilon, \delta, 1$)$^2I)$              \\
                Update global synopsis to $V^{\epsilon\gets\epsilon_i}$ with $V^{\Delta\epsilon}$
                \\
            }
            $\_\text{,} ~V_{A_i}^{\epsilon_i} \gets$ \textsc{additiveGM}($\{ A_i, A_{j, V^{\epsilon} = V^{\epsilon}_{A_j} } \}, \{\epsilon_i, \epsilon \}, q, D$) 
            \\
            Update local synopsis to 
        $V_{A_i}^{\epsilon_i}$ 
                \\
            Update $P[A_i,V]\gets \min (\epsilon, P[A_i,V] + \epsilon_i)$ \\
            \textbf{return} $r_i \gets V_{A_i}^{\epsilon_i}$
        \label{func:aGM_run_end}
    }
    \Fn{\textsc{\emph{privacyTranslate(}$\revise{q_i, v_i}, V, V^{\epsilon}, p$}\emph{)} \label{func:aGM_translation} }{
        $v' \gets$ \textsc{estimateError}($\revise{q_i, v_i}, V^{\epsilon}$) \label{line:aGM_translation_estimate}\\
        $w \gets$ \textsc{minimizer}(func=$-\frac{(\revise{v_i}-w^2v')}{(1-w)^2}$, bounds=$[0, 1]$) \label{line:aGM_translation_optimize} \\
        $\epsilon \gets$ \textsc{privacyTranslate}($q, \frac{(\revise{v_i}-w^2v')}{(1-w)^2}, V, p$) \Comment{C.f. Def \ref{def:analytic_gaussian_translation}.} \label{line:aGM_translation_vanilla} \\
        \textbf{return} $\epsilon$
    \label{func:aGM_translation_end}
    }
    \Fn{\textsc{\emph{constraintCheck(}}$P, A_i, V^\epsilon, \epsilon_i, \Psi$\emph{)} \label{func:aGM_constraint} }{

        $\epsilon' = \min(\epsilon, P[A_i,V]+\epsilon_i) -  P[A_i,V]$\\
        \uIf{$(P$.\textsc{\emph{max(}}axis=Column, $c$=$V_j$\emph{)} +  $ \epsilon' \leq \Psi.\psi_{V_j}) \land$
        $(P$.\textsc{\emph{composite(}}axis=Row, $r$=$P$.\textsc{\emph{max}}\emph{())} + $\epsilon' \leq \Psi.\psi_{P}) \land$
        $(P$.\textsc{\emph{composite(}}axis=Row\emph{)} +  $ \epsilon' \leq \Psi.\psi_{A_i})$ }{
            \textbf{return} True/Pass
        }
        
        \label{func:aGM_constraint_end}
    }
    \caption{Additive Gaussian Approach}
    \label{alg:additive_gaussian}
\end{algorithm}

\subsubsection{Provenance Constraint Checking}
The provenance constraint checking for the additive Gaussian approach (line \ref{func:aGM_constraint}) is similar to the module for the \baseline approach.
We would like to highlight 3 differences. %
1) Due to the use of the additive Gaussian mechanism, the composition across analysts on the same view is bounded as tight as the $\max P[A_i, V], \forall A_i \in \mc{A}$.
Therefore, we check the view constraint by taking the \textsc{max} over the column retrieved by the index of view $V_j$.
2) To check table constraint, we composite a vector where each element is the maximum recorded privacy budget over every column/view.
3) The new cumulative cost is no longer $\epsilon_i$, but $\epsilon' = \min(\epsilon, P[A_i,V]+\epsilon_i) -  P[A_i,V]$.

\begin{theorem}
Given the same sequence of queries $Q$ and at least 2 data analysts in the system and the same system setup, the total number of queries answered by the additive Gaussian approach is always more than or equal to that answered by \baseline approach.
\end{theorem}

\ifarxiv
\begin{proof}[Proof Sketch]
    \oursystem processes incoming queries with DP synopsis (\baseline approach) or local synopsis (additive Gaussian approach). For each query $q$ in $Q$, if it can be answered (w.r.t the accuracy requirement) with cached synopsis, both approaches will process it in the same manner; otherwise, \oursystem needs to update the synopses.
    Comparing the cost of synopses update in both methods, $\min(\epsilon, P[A_i,V]+\epsilon_i) -  P[A_i,V] \leq \epsilon_i$ always holds.
    Therefore, given the same privacy constraints, if a synopsis can be generated to answer query $q$ with \baseline approach, the additive Gaussian approach must be able to update a global synopsis and generate a local synopsis to answer this query, which proves the theorem.
\end{proof}

We note that \baseline approach generates independent synopses for different data analysts, while in the additive Gaussian approach, we only update global synopses which saves privacy budgets when different analysts ask similar queries. We empirically show the benefit in terms of answering more queries with the additive Gaussian approach in Section~\ref{sec:evaluation}. 
\fi

\subsubsection{Putting Components All Together}
The additive Gaussian approach is presented in Algorithm \ref{alg:additive_gaussian}: \ref{func:aGM_run}-\ref{func:aGM_run_end} (the function \textsc{run}).
At each time stamp, the system receives the query $q$ from the analyst $A_i$ and selects the view that this query can be answerable.
If the translated budget $\epsilon_i$ is greater than the budget allocated to the global synopsis of that view (line 3), we generate a delta synopsis (lines 4-5) and update the global synopsis (line 6).
Otherwise, we can use additive GM to generate the local synopsis based on the (updated) global synopsis (lines 7-9).
We update the provenance table with the consumed budget $\epsilon_i$ (line 10) and answer the query based on the local synopsis to the analyst (line 11).

\subsubsection{Discussion on Combining Local Synopses}
We may also update a local synopsis $V_{A_i}^{\epsilon'}$ (upon request $\epsilon_i$) by first generating an intermediate local synopsis $V_{A_i}^{\Delta\epsilon}$ from $V^{\epsilon_t}$ using additive GM, where $\Delta\epsilon=\epsilon_i - \epsilon'$ and
then combining $V_{A_i}^{\Delta\epsilon}$ with the previous local synopsis in a similar way as it does for the global synopses, which leads to a new local synopsis 
$V_{A_i}^{\epsilon'+\Delta\epsilon}$. 
\begin{example}
    To answer Bob's query $(q_1, \epsilon_{V_1, Bob} = 0.7)$, we can use additive GM to generate a fresh local synopsis $V^{0.4}_{Bob}$ from $V^{0.7}$, and combine $V^{0.4}_{Bob}$ with the existing $V^{0.3}_{Bob}$ to get $V^{0.7}_{Bob}$.
\end{example}

However, unlike combining global synopses, these local synopses share correlated noise. We must solve a different optimization problem to find an unbiased estimator with minimum error. For example, given the last combined global synopsis (and its weights) $V' = w_{t-1}V^{\epsilon_{t-1}} + w_t V^{\epsilon_t - \epsilon_{t-1}}$, if we know $V_{A}^{\epsilon_{t-1}}$ is a fresh local synopsis generated from $V^{\epsilon_{t-1}}$, we consider using the weights $k_{t-1}, k_t$ for local synopsis combination:
\begin{align}
    V_{A}' = & ~ k_{t-1} V_{A}^{\epsilon_{t-1}} + k_t V_{A}^{\Delta \epsilon} \nonumber \\
    = & ~ k_{t-1} (V^{\epsilon_{t-1}} + \eta_{A}^{t-1}) + k_t (w_{t-1}V^{\epsilon_{t-1}} + w_t V^{\epsilon_t - \epsilon_{t-1}} + \eta_{A}^{t}) \nonumber \\ 
    = & ~ (k_{t-1} + k_t w_{t-1}) V^{\epsilon_{t-1}} + k_t w_{t} V^{\epsilon_t - \epsilon_{t-1}} + 
    k_{t-1}\eta_{A}^{t-1}+k_t\eta_{A}^t,
    \nonumber 
\end{align}
where $\eta_{A}^{t-1}$ and $\eta_{A}^{t}$ are the noise added to the local synopses in additive GM with variance $\sigma_{t-1}^2$ and $\sigma_t^2$. We can find the adjusted weights that minimize the expected error for $V_{A}'$ is $v_{A, t} = (k_{t-1} + k_t w_{t-1})^2 v_{t-1} + k_t^2 w_{t}^2 v_{\Delta} + k_{t-1}^2\sigma_{t-1}^2 + k_t^2 \sigma_t^2$ subject to $k_{t-1} + k_t w_{t-1}+k_t w_{t} =1$, using an optimization solver. 
Allowing the optimal combination of local synopses tightens the cumulative privacy cost of $A_i$ on $V$, i.e., $\epsilon_i < \min (\epsilon_t, P[A_i, V]+ \epsilon_i)$.
However, if the existing local synopsis $V_{A}^{\epsilon_{t-1}}$ is a combined synopsis from a previous time stamp, the correct variance calculation requires a nested analysis on \emph{from where the local synopses are generated and with what weights the global/local synopses are combined}.
This renders too many parameters for \oursystem to keep track of and puts additional challenges for accuracy translation since solving an optimization problem with many weights is hard, if not intractable.
For practical reasons, \oursystem adopts the approach described in Algorithm \ref{alg:additive_gaussian}.

\subsection{Configuring Provenance Table}
\label{subsec:algo_col_constraints}

In existing systems, the \admin is only responsible for setting a single parameter specifying the overall privacy budget.
\oursystem, however, requires the \admin to configure more parameters.

\subsubsection{Setting Analyst Constraints for Proportional Fairness}

We first discuss guidelines to establish per-analyst constraints in \oursystem, by proposing two specifications that achieve proportional fairness.

\begin{definition}[Constraints for \Baseline Approach]
\label{def:vanilla_analyst_constraint}
    For the \baseline approach, we propose to specify each analyst's constraint by proportional indicated normalization.
    That is, for the table-level constraint $\psi_P$ and analyst $A_i$ with privilege $l_i$, %
    $\psi_{A_i} = \frac{l_i}{\sum_{j\in [n]} l_j} \psi_P$.
\end{definition}

Note that this proposed specification is not optimal for the additive Gaussian approach, since the maximum utilized budget will then be constrained by $\max \psi_{A_i} < \psi_P$ when more than 1 analyst is using the system.
Instead, we propose the following specification.

\begin{definition}[Constraints for Additive Gaussian]
\label{def:agm_analyst_constraint}
    For the additive Gaussian approach, each analyst $A_i$'s constraint can be set to as $\frac{l_i}{l_{max}}\psi_P$, where $l_{max}$ denotes the maximum privilege in the system.
\end{definition}

\stitle{Comparing Two Specifications.}
We compare the 
two analyst-constraint specifications (Def. \ref{def:vanilla_analyst_constraint} and \ref{def:agm_analyst_constraint}) 
with experiments in Section \ref{subsec:results}.
Besides their empirical performance, the \baseline constraint specification (Def. \ref{def:vanilla_analyst_constraint}) requires all data analysts to be registered in the system before the provenance table is set up.
The additive Gaussian approach (with specification in Def. \ref{def:agm_analyst_constraint}), however, allows for the inclusion of new data analysts at a later time.

\subsubsection{Setting View Constraints for Dynamic Budget Allocation}

Existing privacy budget allocator \cite{kotsogiannis2019privatesql} adopts a static splitting strategy such that the per view privacy budget is equal or proportional to their sensitivity.
\oursystem subsumes theirs by, similarly, setting the view constraints to be equal or proportional to the sensitivity of each view, i.e., $\{\psi_{V_j}|V_j\in \mc{V} \} = \{ \lambda_{V_j} \cdot \epsilon/\hat{\Delta}_{V_j} \}_{\forall V_j \in \mc{V}}$, where $\hat{\Delta}_{V_j}$ is the upper bound of the sensitivity of the view $V_j$.
We therefore propose the following water-filling view constraint specification for a better view budget allocation.

\begin{definition}[Water-filling View Constraint Setting]
\label{def:view_constraint_water_filling}
The table constraint $\psi_P$ has been set up as a constant (i.e., the overall privacy budget) in the system.
The \admin simply set all view constraints the same as the table constraint, $\psi_{V_j} \coloneqq \psi_P, \forall V_j \in \mc{V}$. 
\end{definition}

With water-filling constraint specification, the provenance constraint checking will solely depend on the table constraint and the analyst constraints.
The overall privacy budget is then dynamically allocated to the views based on analysts' queries.
Compared to existing budget allocation methods on views, our water-filling specification based on the \provT reflects the actual accuracy demands on different views.
Thus it results in avoiding the waste of privacy budget and providing better utility:
i) \oursystem spends fewer budgets on unpopular views, whose consumed budget is less than $\lambda_{V_j} \cdot \epsilon/\hat{\Delta}_{V_j}$;
ii) \oursystem can answer queries whose translated privacy budget $\epsilon > \lambda_{V_j} \cdot \epsilon/\hat{\Delta}_{V_j}$.
In addition, the water-filling specification allows \oursystem to add views over time.

\subsection{Privacy and Fairness Guarantees}
\label{sec:guarantees}

\begin{theorem}[System Privacy Guarantee]
\label{thm:privacy_guarantee}
Given the \provT and its constraint specifications,   
$\Psi=\{\psi_{A_i}|A_i\in \mc{A}\}\cup \{\psi_{V_j}|V_j\in \mc{V} \}\cup \{\psi_P \}$, both mechanisms for \oursystem ensure $[\ldots, (A_i, \psi_{A_i},\delta),\ldots]$-multi-analyst-DP; they also ensure $\min (\psi_{V_j}, \psi_P)$-DP for view $V_j\in \mc{V}$ and overall $\psi_P$-DP if all the data analysts collude. 
\end{theorem}

With the provenance table, \oursystem can achieve proportional fairness when analysts submit a sufficient number of queries, stated as in the following theorem.
\ifarxiv
Both proofs are deferred to appendices.
\fi

\begin{theorem}[System Fairness Guarantee]
\label{thm:fairness_guarantee}
Given the \provT $P$ and the described approaches of setting analyst constraints, both mechanisms achieve proportional fairness, when the data analysts finish consuming their assigned privacy budget.
\end{theorem}

\section{Experimental Evaluation}
\label{sec:evaluation}

In this section, 
we compare \oursystem with baseline systems and conduct an ablation investigation for a better understanding of different components in \oursystem. 
Our goal is to show that \oursystem can improve existing query answering systems for multi-analyst in terms of the number of queries answered and fairness.

\subsection{Experiment Setup}
We implement \oursystem in Scala with PostgreSQL for the database system, and 
deploy it as a middle-ware between Chorus \cite{johnson2020chorus} and multiple analysts. 
\revise{\linelabel{line:delta_exp}Our implementation follows existing practice \cite{johnson2020chorus} to set all $\delta$ parameter to be the same and small (e.g., 1e-9) for all queries.
The $\delta$ parameters in the privacy constraints (column/row/table) in the privacy provenance table are set to be capped by the inverse of the dataset size.}

\subsubsection{Baselines}
Since we build on top of Chorus~\cite{johnson2020chorus}, we develop a number of baseline approaches for the multi-analyst framework using Chorus as well for fair comparisons.

\squishlist
\item  \textit{Chorus \cite{johnson2020chorus}:} This baseline is the plain Chorus, which uses GM and makes no distinction between data analysts and uses no views. The system only sets the overall privacy budget. %

\item \textit{\oursystem minus Cached Views (ChorusP):} This baseline enables privacy provenance tracking for each data analyst but does not store any synopses. 
We use Def. \ref{def:vanilla_analyst_constraint} to set up analyst constraints and Def. \ref{def:view_constraint_water_filling} for view constraints.

\item \textit{\oursystem minus Additive GM (Vanilla):} We equip Chorus with \provT and the cached views, but with our vanilla approach to update and manage the provenance table and the views. 
The \provT is configured the same as in ChorusP.

\item \textit{Simulating PrivateSQL  (sPrivateSQL)\cite{kotsogiannis2019privatesql}:} 
We simulate PrivateSQL by generating the static DP synopses at first. The budget allocated to each view is proportional to the view sensitivities~\cite{kotsogiannis2019privatesql}. Incoming queries that cannot be answered accurately with these synopses  will be rejected.
\squishend

\begin{figure*}[t]
    \centering
    \includegraphics[width=\textwidth]{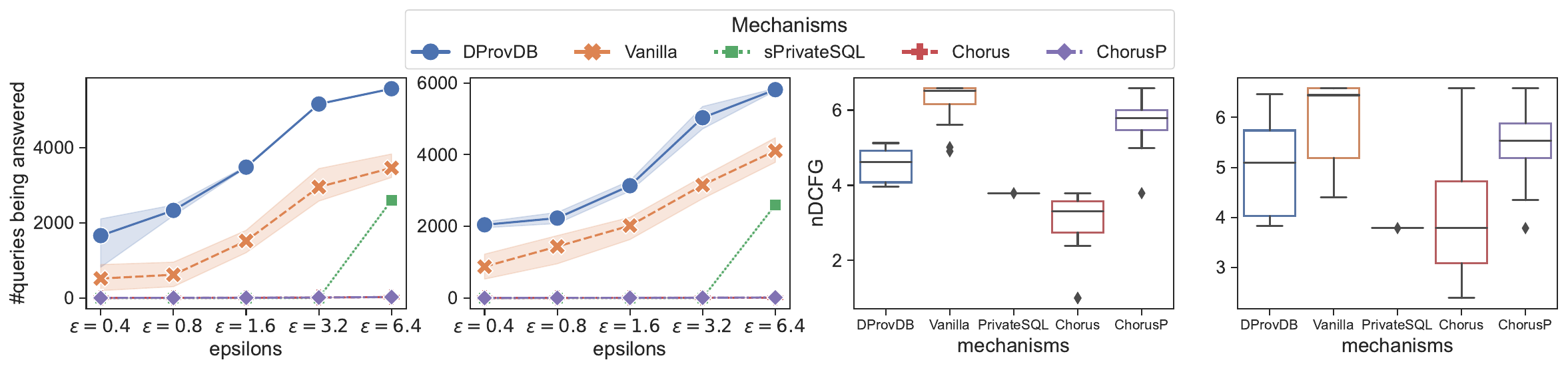}
    \caption{End-to-end Comparison (\emph{RRQ} task, over \emph{Adult} dataset), from left to right: a) utility v.s. overall budget, round-robin; b) utility v.s. overall budget, randomized; c) fairness against baselines, round-robin; d) fairness against baselines, randomized. 
    }
    \label{fig:end_to_end_RRQ_adult}
\end{figure*}

\subsubsection{Datasets and Use Cases} We use the Adult dataset \cite{uci_dataset} (a demographic data of 15 attributes and 45,224 rows) and the TPC-H dataset (a synthetic data of 1GB)~\cite{tpch_dataset} for experiments.
We consider the following use cases. %

\squishlist
\item \textit{Randomized range queries (RRQ):}
We randomly generate 4,000 range queries \emph{per analyst}, each with one attribute randomly selected with bias.
Each query has the range specification $[s, s+o]$ where $s$ and the offset $o$ are drawn from a normal distribution. 
We design two types of query sequences from the data analysts: 
a) \textbf{round-robin}, where the analysts take turns to ask queries; 
b) \textbf{random}, where a data analyst is randomly selected each time. 

\item \textit{Breadth-first search (BFS) tasks:} Each data analyst explores a dataset by traversing a decomposition tree of the cross product over the selected attributes, aiming to find the (sub-)regions with underrepresented records. That is, the data analyst traverses the domain and terminates only if the returned noisy count is within a specified threshold range.

\squishend
To answer the queries, we generate one histogram view on each attribute.
We use two analysts with privileges 1 and 4 by default setting and also study the effect of involving more analysts.

\subsubsection{Evaluation Metrics}
We use four metrics for evaluation.

\squishlist
\item \textit{Number of queries being answered}: We report the number of queries that can be answered by the system when
no more queries can be answered as the utility metric to the system. %

\item \textit{Cumulative privacy budget}: 
For BFS tasks that have fixed workloads, it is possible that the budget is not used up when the tasks are complete. Therefore, we report the total cumulative budget consumed by all data analysts when the tasks end.  

\item \textit{Normalized discounted cumulative fairness gain (nDCFG)}: We coined an empirical fairness measure here. First, we introduce DCFG measure for a mechanism $\mc{M}$ as
$\text{DCFG}_\mc{M} = \sum^n_{i=1} \frac{| Q_{A_i} |}{\log_2 (\frac{1}{l_i} +1)}$, where 
$l_i$ is the privilege level of analyst $A_i$ and $| Q_{A_i} |$ is the total number of queries of $A_i$ being answered. Then nDCFG is DCFG normalized by the total number of queries answered. 

\item \textit{Runtime:} We measures the run time in milliseconds.

\squishend

We repeat each experiment 4 times using different random seeds. 

\begin{figure}[t]
    \centering
    \includegraphics[width=\linewidth]{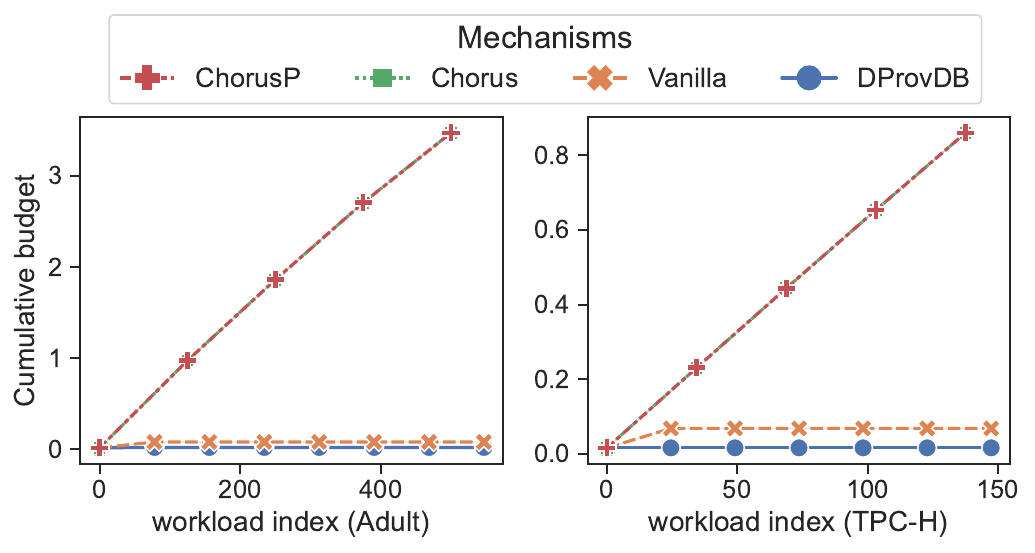}
    \caption{End-to-end Comparison (BFS task): Cumulative privacy budget assumption v.s. workload indices. Left: over \textit{Adult} dataset; Right: over \textit{TPC-H} dataset. 
    }
    \label{fig:e2e_bfs}
\end{figure}

\subsection{Empirical Results}
\label{subsec:results}

We initiate experiments on the default settings of \oursystem and the baseline systems for comparison, and then study the effect of modifying the components of \oursystem.

\subsubsection{End-to-end Comparison}
\label{subsec:exp_end_to_end}

This comparison is with the setup of the
analyst constraints in line with Def. \ref{def:agm_analyst_constraint} for \oursystem, and with Def. \ref{def:vanilla_analyst_constraint} for \baseline approach. We brief our findings.

\stitle{Results of RRQ task.}
We present Fig. \ref{fig:end_to_end_RRQ_adult} for this experiment.
We fix the entire query workload and run the query processing on the five mechanisms.
\oursystem outperforms all competing systems, in both round-robin or random use cases.
Chorus and ChorusP can answer very few queries because their privacy budgets are depleted quickly.
Interestingly, our simulated PrivateSQL can answer a number of queries which are comparable to the \baseline appraoch when $\epsilon=6.4$, but answers a limited number of queries under higher privacy regimes.
This is intuitive because if one statically fairly pre-allocates the privacy budget to each view when the overall budget is limited, then possibly every synopsis can hardly answer queries accurately.
One can also notice the fairness score of ChorusP is significantly higher than Chorus, meaning the enforcement of \provT can help achieve fairness.

\stitle{Results of BFS task.}
An end-to-end comparison of the systems on the BFS tasks is depicted in Fig. \ref{fig:e2e_bfs}.
Both \oursystem and \baseline can complete executing the query workload with near constant budget consumption, while Chorus(P) spends privacy budget linearly with the workload size.
\oursystem saves even more budget compared to \baseline approach, based on results on the TPC-H dataset.

\stitle{Runtime performance.}
Table \ref{tab:tpc_run_time_performance} summarizes the runtime of \oursystem and all baselines. \oursystem and sPrivateSQL %
use a rather large amount of time to set up views, however, answering large queries based on views saves more time on average compared to Chorus-based mechanisms.
Recall that aGM requires solving optimization problems, where empirical results show incurred time overhead is only less than 2 ms per query.

We draw the same conclusion for our RRQ experiments and runtime performance test on the other dataset. 
\ifarxiv
We defer the results to the appendices.
\fi

\begin{figure*}[t]
    \centering
    \includegraphics[width=\textwidth]{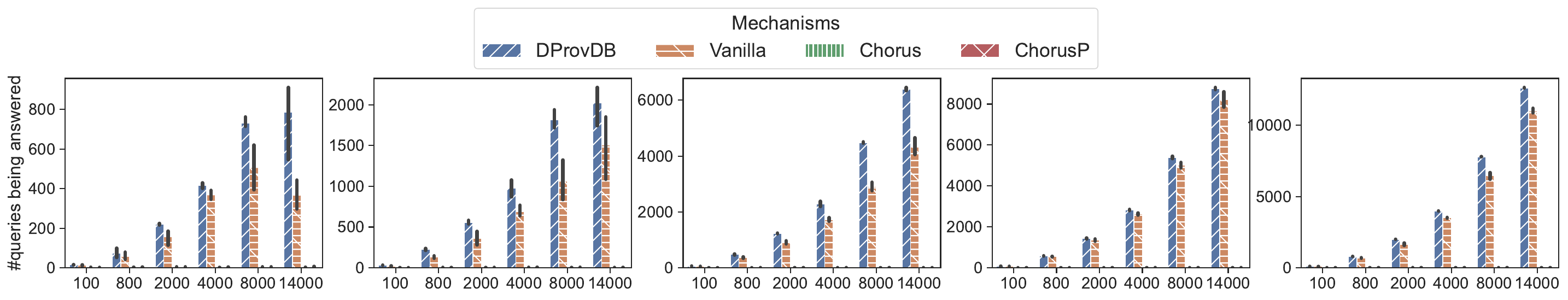}
    \caption{Component Comparison (\emph{RRQ} task, over \emph{Adult} dataset): Enabling Cached Synopses. Utility v.s. the size of query workload (round-robin)From left to right: $\epsilon = \{0.4, 0.8, 1.6, 3.2, 6.4\}$. 
        }
    \label{fig:caches_RRQ_adult}
\end{figure*}

\begin{figure}[t]
    \centering
    \includegraphics[width=\linewidth]{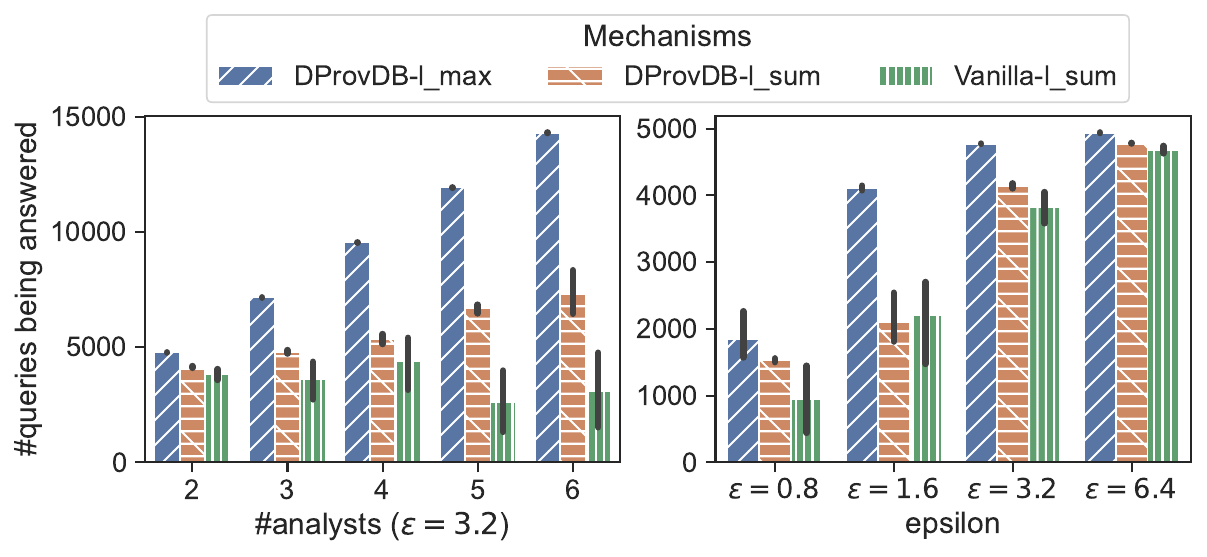}
    \caption{Component Comparison (\emph{RRQ} task, over \emph{Adult} dataset): Additive GM v.s. \Baseline. Left: utility v.s. \#analysts, round-robin; Right: utility v.s. overall budgets, round-robin.
    }
    \label{fig:agm_RRQ_adult}
\end{figure}

\begin{table}[tbp]
    \centering
    \caption{Runtime Performance Comparison over Different Mechanisms on TPC-H Dataset (running on a Linux server with 64 GB DDR4 RAM and AMD Ryzen 5 3600 CPU, measured in milliseconds)}
    \resizebox{1.0\linewidth}{!}{%
    \begin{tabular}{ccccc}
    \toprule \hline
    Systems  & Setup Time   & Running Time & No. of Queries & Per Query Perf  \\ \hline
    \oursystem     &  20388.65 ms    &  297.30 ms  & 86.0  & 3.46 ms \\ 
    Vanilla &  20388.65 ms    & 118.04 ms  & 86.0   & 1.37 ms  \\ 
    sPrivateSQL &  20388.65 ms &  166.51 ms & 86.0 & 1.94 ms \\ 
    Chorus &  N/A   &  7380.47 ms  & 62.0   & 119.04 ms \\ 
    ChorusP & N/A   &  7478.69 ms  & 62.0  & 120.62 ms \\ \hline
    \bottomrule
\end{tabular}
    }
    \label{tab:tpc_run_time_performance}
\end{table}

\subsubsection{Component Comparison}
\label{subsec:exp_comp}

We consider three components in \oursystem to evaluate separately.

\smallskip

\stitle{Cached Synopses.}
Given the same overall budget, mechanisms leveraging cached synopses outperform those without caches in terms of utility when the size of the query workload increases.
This phenomenon can be observed for all budget settings, as shown in Fig. \ref{fig:caches_RRQ_adult}.
Due to the use of cached synopses, \oursystem can potentially answer more queries if the incoming queries hit the caches.
Thus, the number of queries being answered would increase with the increasing size of the workload, given the fixed overall budget.

\stitle{Additive GM v.s. \Baseline.}
Given the same overall budget, additive GM outperforms the \baseline mechanism on utility. 
The utility gap increases with the increasing number of data analysts presented in the system.
The empirical results are presented in Fig. \ref{fig:agm_RRQ_adult}.
When there are 2 analysts in the system, additive GM can only gain a marginal advantage over \baseline approach; when the number of data analysts increases to 6, additive GM can answer around $\sim$2-4x more queries than \baseline.
We also compare different settings of analyst constraints among different numbers of analysts and different overall system budgets (with 2 analysts).
It turns out that additive GM with the setting in Def. \ref{def:agm_analyst_constraint} (DProvDB-l\_max), is the best one, outperforming the setting from Def. \ref{def:vanilla_analyst_constraint} (DProvDB-l\_sum, Vanilla-l\_sum) all the time for different epsilons and with $\sim$4x more queries being answered when \#analysts=6.

\stitle{Constraint Configuration.}
If we allow an expansion parameter $\tau \geq 1$ on setting the analyst constraints (i.e., overselling privacy budgets than the constraint defined in Def.~\ref{def:agm_analyst_constraint}), we can obtain higher system utility by sacrificing fairness, as shown in 
 Fig. \ref{fig:fairness_RRQ_adult}.
With 2 data analysts, the number of total queries being answered by additive GM is slightly increased when we gradually set a larger analyst constraint expansion rate.
Under a low privacy regime (viz., $\epsilon=3.2$), this utility is increased by more than 15\% comparing setting $\tau=1.9$ and $\tau=1.3$;
on the other hand, the fairness score decreases by around 10\%.

This result can be interpretable from a basic economic view, that we argue, as a system-wise public resource, the privacy budget could be \emph{idle resources} when some of the data analysts stop asking queries.
Thus, the \emph{hard privacy constraints} enforced in the system can make some portion of privacy budgets unusable.
Given the constraint expansion, we allow \oursystem to tweak between a fairness and utility trade-off, while the overall privacy is still guaranteed by the table constraint.

\begin{figure}[t]
    \centering
    \includegraphics[width=\linewidth]{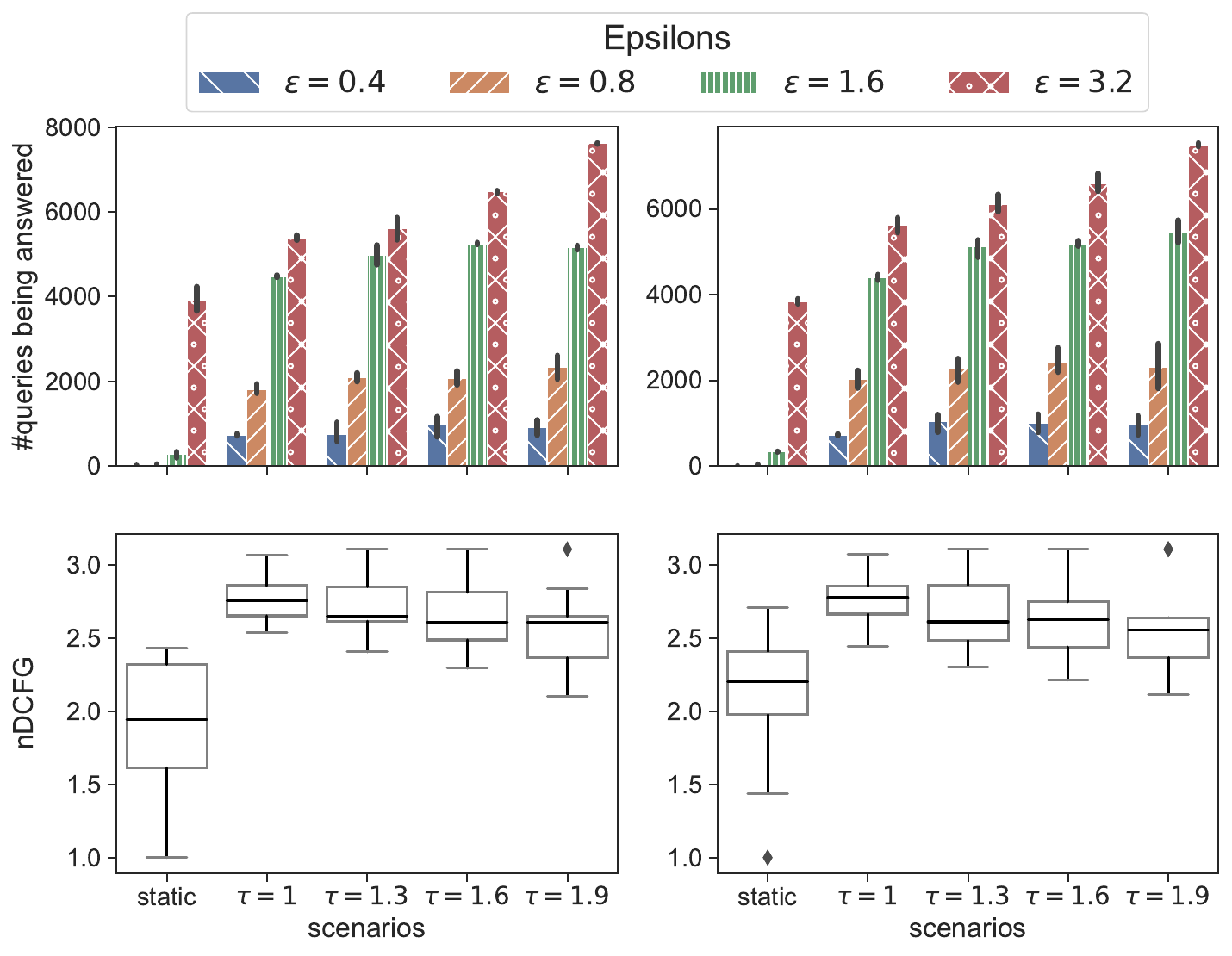}
    \caption{Component Comparison (\emph{RRQ} task, over \emph{Adult} dataset): Constraint Configuration. First row: utility v.s. constraint settings. Second row: fairness v.s. constraint settings. Left: round-robin. Right: randomized. 
    }
    \label{fig:fairness_RRQ_adult}
\end{figure}

\revise{
\stitle{Varying $\delta$ parameter.} \linelabel{line:experiment_delta}
In this experiment, we control the overall privacy constraint $\epsilon=6.4$ and vary the $\delta$ parameter as per query.
We use the BFS workload as in our end-to-end experiment and the results are shown in Fig.~\ref{fig:delta_paper}. 
While varying a small $\delta$ parameter does not much affect the number of queries being answered, observe that increasing $\delta$, \oursystem can slightly answer more queries. This is because, to achieve the same accuracy requirement, the translation module will output a smaller $\epsilon$ when $\delta$ is bigger, which will consume the privacy budget slower.
Note that the overall privacy constraint on $\delta$ should be set no larger than the inverse of the dataset size.
Setting an unreasonably large per-query $\delta$ can limit the total number of queries being answered.

\stitle{Other experiments.} \linelabel{line:experiment_other}
We also run experiments to evaluate \oursystem on data-dependent utility metric, i.e. relative error \cite{XiaoBHG11relativeerror}, and empirically validate the correctness of the accuracy-privacy translation module. 
\ifarxiv
In particular, we performed experiments to show that the noise variance $v_q$ (derived from the variance of the noisy synopsis) of the query answer,  according to the translated privacy budget, is always no more than the accuracy requirement $v_i$ submitted by the data analyst.
As shown in Fig. \ref{fig:correctness_of_translation} (a), the difference between the two values, $v_q - v_i$, is always less than 0, and very small for a given BFS query workload (where the accuracy requirement is set to be above $v_i > 10000$).

Furthermore, we consider the following data-dependent utility, namely relative error \cite{XiaoBHG11relativeerror}, which is defined as 
$$\text{Relative Error} = \frac{|\text{True Answer}-\text{Noisy Answer}|}{\max \{\text{True Answer}, c\} },$$
where $c$ is a specified constant to avoid undefined values when the true answer is 0. 

Note that \oursystem does not specifically support analysts to submit queries with \emph{data-dependent} accuracy requirements.
The translated query answer can have a large relative error if the true answer is small or close to zero.
We thereby merely use this utility metric to empirically evaluate the answers of a BFS query workload, as a complementary result [Fig.~\ref{fig:correctness_of_translation} (b)] to the paper. 
\oursystem and the Vanilla approach have a larger relative error than Chorus and ChorusP because they can answer more queries, many of which have a comparatively small true answer -- incurring a large relative error, than Chorus-based methods.
\fi
}

\begin{figure}[t]
    \centering
    \includegraphics[width=\linewidth]{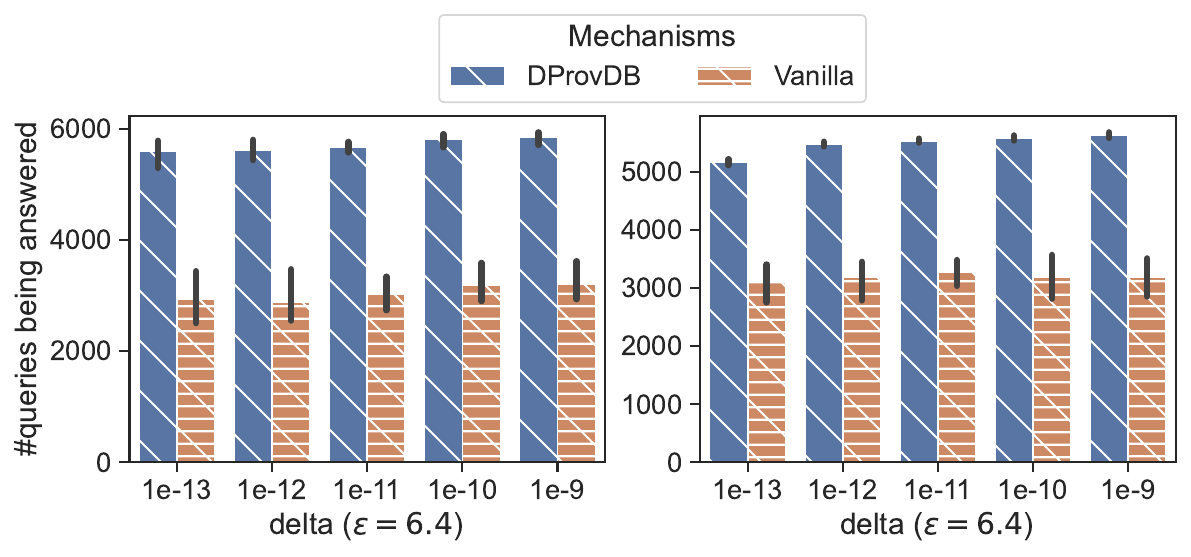}
    \caption{\revise{\#queries being answered vs. varying delta parameter (BFS task,  \textit{Adult}). Left: round-robin, Right: randomized.}}
    \label{fig:delta_paper}
\end{figure}

\ifarxiv
\begin{figure}
    \centering
    \includegraphics[width=0.22\textwidth]{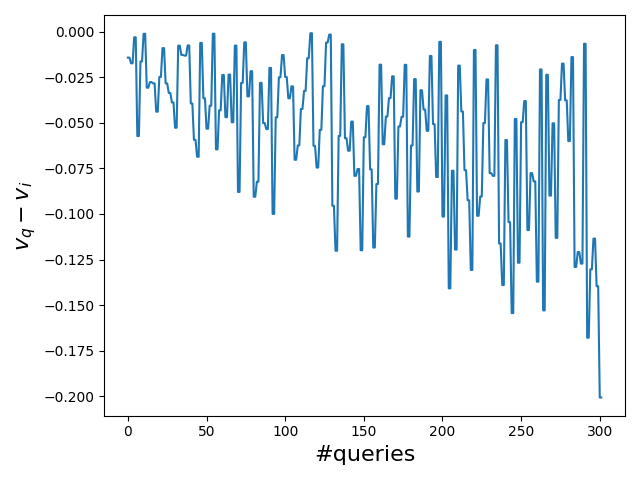}
    \includegraphics[width=0.21\textwidth]{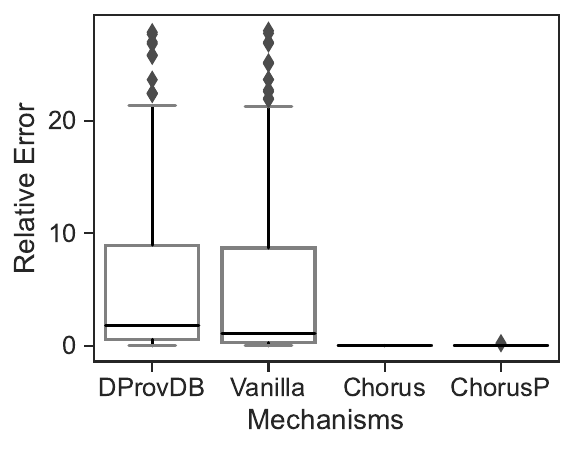}
    \caption{\revise{(a) The cumulative average of $v_q - v_i$ for a BFS query workload (on \textit{Adult} dataset), where $v_i$ represents the submitted accuracy requirement and $v_q$ denotes the noise variance of the query answer. (b) Relative error of processing the BFS query workload (on \textit{Adult} dataset) among different mechanisms.}}
    \label{fig:correctness_of_translation}
\end{figure}
\fi

\ifarxiv
\section{Discussion}
In this section, we discuss a weaker threat model of the multi-analyst corruption assumption, with which additional utility gain is possible.
We also discuss other strawman solutions toward a multi-analyst DP system.
\fi

\label{sec:discussion_new}

\ifarxiv
\subsection{Relaxation of Multi-analyst Threat Model}
\fi

So far we have studied that all data analysts can collude.
A more practical setting is, perhaps, considering a subset of data analysts that are compromised by the adversary.
This setting is common to see in the multi-party computation research community \cite{dwork2006our} (\emph{a.k.a} active security \cite{damgaard2007scalable,beaver1989multiparty}), where $t$ out of $n$ participating parties are assumed to be corrupted.

\begin{definition}[$(t, n)$-compromised Multi-analyst Setting]
    We say a multi-analyst setting is $(t, n)$-compromised, if there exist $n$ data analysts where at most $t$ of them are  \textbf{malicious} (meaning they can collude in submitting queries and sharing the answers).
\end{definition}

\ifarxiv
The $(t ,n)$-compromised multi-analyst setting makes weaker assumptions on the attackers. 
Under this setting, the privacy loss is upper bounded by $(\sum_t \epsilon_i, \sum_t \delta_i)$, which is the summation over $t$ largest privacy budgets.
However, we cannot do better than the current \oursystem algorithms with this weaker setting under worst-case privacy assumption.
\fi

\begin{theorem}[Hardness on Worst-case Privacy]
Given a mechanism $\mc{M}$ which is $[\ldots, (A_i, \epsilon_i, \delta_i), \dots]$-multi-analyst-DP, the worst-case privacy loss under $(t, n)$-compromised multi-analyst DP is lower bounded by $(\max \epsilon_i, \max \delta_i)$, which is the same as under the all-compromisation setting.
\end{theorem}

\begin{proof}[Proof Sketch]
Under the worst-case assumption, the analyst with the largest privacy budget $(\max \epsilon_i, \max \delta_i)$ is sampled within the $t$ compromised data analysts.
Then it is natural to see that the lower bound of the privacy loss is $(\max \epsilon_i, \max \delta_i)$.
\end{proof}

At first glance, the relaxation of the $(t, n)$-compromisation does not provide us with better bounds.
A second look, however, suggests the additional trust we put in the data analyst averages the privacy loss in compromisation cases.
Therefore, with this relaxed privacy assumption, it is possible to design mechanisms to achieve better utility using policy-driven privacy.

\stitle{Policies for multi-analyst.}
Blowfish privacy \cite{he2014blowfish} specifies different levels of protection over sensitive information in the curated database.
In the spirit of Blowfish, we can use policies to specify different levels of trust in data analysts using \oursystem.

\begin{definition}[$(t, n)$-Analysts Corruption Graph]
    Given $n$ data analysts and assuming $(t, n)$-compromised setting, we say an undirected graph $G = (V, E)$ is a $(t, n)$-analysts corruption graph if, 
    \begin{itemize}
        \item Each node in the vertex set $v_i \in V$ represents an analyst $A_i$;

        \item An edge is presented in the edge set $\e(v_i, v_j) \in E$ if data analysts $A_i$ and $A_j$ can collude;

        \item Every connected component in $G$ has less than $t$ nodes.
    \end{itemize}
        
\end{definition}

The corruption graph models the prior belief of the policy designer (or DB administrator) to the system users.
Groups of analysts are believed to be not compromised if they are in the disjoint sets of the corruption graph.
Based on this corruption graph, we can specify the analysts' constraints as assigning the overall privacy budget $\psi_P$ to each connected component. 

\begin{theorem}
\label{thm:corruption_graph}
    There exist mechanisms with $(t, n)$-multi-analyst DP perform at least as well as with multi-analyst DP.
\end{theorem}

\ifarxiv
\begin{proof}
    Given a $(t, n)$-analysts corruption graph $G$, we show a na\"ive budget/constraint that makes $(t, n)$-multi-analyst DP degrade to multi-analyst DP.
    Ignoring the graph structure, we split the overall privacy budget $\psi_P$ and assign a portion to each node proportional to the privilege weight on each node.
\end{proof}
\fi

Let $k$ be the number of disjoint connected components in $G$. 
Then we have at most $\lceil k \rceil \cdot \psi_P$ privacy budget to assign to this graph.
Clearly, the mechanisms with $(t, n)$-multi-analyst DP achieve $(\lceil k \rceil \cdot \psi_P)$-DP.
When $n > t$, we have more than 1 connected component, meaning the overall privacy budget we could spend is more than that in the all-compromisation scenario.

\ifarxiv
\subsection{Comparison with Strawman Solutions}

One may argue for various alternative solutions to the multi-analyst query processing problem, as opposed to our proposed system. 
We justify these possible alternatives and compare them with \oursystem.

\stitle{Strawman \#1: Sharing Synthetic Data.} 
Recall that \oursystem generates global and local synopses from which different levels of noise are added with additive GM, and answers analysts' queries based on the local synopses.
We show our proposed algorithm is optimal in using privacy budgets when all data analysts collude.
One possible alternative may be just releasing the global synopses (or generating synthetic data using all privacy budgets) to every data analyst, which also can achieve optimality in the all-compromisation setting.
We note that this solution is $\min \left(\psi_P, (\max \epsilon_i, \max \delta_i)\right)$-DP (same as the overall DP guarantee provided by \oursystem), however, it does not achieve the notion of \emph{Multi-analyst DP} (all data analysts get the same output).

\stitle{Strawman \#2: Pre-calculating Seeded Caches.}
To avoid the cost of running the algorithm in an online manner, one may consider equally splitting the privacy budgets into small ones and using additive GM to pre-compute all the global synopses and local synopses.
This solution saves all synopses and for future queries, it can directly answer by finding the appropriate synopses.
If the space cost is too large, alternatively one may store only the seeds to generate all the synopses from which a synopsis can be quickly created.

This scheme arguably achieves the desired properties of \oursystem, if one ignores the storage overhead incurred by the pre-computed synopses or seeds.
However, for an online query processing system, it is usually unpredictable what queries and accuracy requirements the data analysts would submit to the system. 
This solution focuses on doing most of the calculations offline, which may, first, lose the precision in translating accuracy to privacy, leading to a trade-off between precision and processing time regarding privacy translation.
We show in experiments that the translation only happens when the query does not hit the existing cache (which is not too often), and the per-query translation processing time is negligible.
Second, to pre-compute all the synopses, one needs to pre-allocate privacy budgets to the synopses.
We has shown as well using empirical results that this approach achieves less utility than \oursystem.

\fi

\section{Related Work}
\label{sec:related_work}

Enabling DP in query processing is an important line of research in database systems~\cite{near2021differential}.
Starting from the first end-to-end interactive DP query system \cite{mcsherry2009privacy}, existing work has been focused on generalizing to a larger class of database operators \cite{proserpio2014calibrating,johnson2018towards,dong2021residual,dong2022r2t,amin2022plume}, building a programmable toolbox for experts \cite{johnson2020chorus, zhang2018ektelo}, or providing a user-friendly accuracy-aware interface for data exploration \cite{ge2019apex,lobo2020programming}.
Another line of research investigated means of saving privacy budgets in online query answering, including those based on cached views \cite{kotsogiannis2019privatesql,mazmudar2022cache} and the others based on injecting correlated noise to query results \cite{xiao2022answering,koufogiannis2015gradual,bao2022skellam}.
Privacy provenance (on protected data) is studied in the personalized DP framework \cite{ebadi2015differential} as a means to enforce varying protection on database records, which is dual to our system. 
The multi-analyst scenario is also studied in DP literature~\cite{xiao2021optimizing,knopf2021framework,pujol2021budget,pujol2022multi}. 
Our multi-analyst DP work focuses on the online setting, which differs from the offline setting, where the entire query workload is known in advance~\cite{xiao2021optimizing,knopf2021framework,pujol2021budget}. The closest related work~\cite{pujol2022multi} considers an online setup for multi-analyst DP, but the problem setting differs. The multi-analyst DP in~\cite{pujol2022multi} assumes the data analysts have the incentive to share their budgets to improve the utility of their query answers, but our multi-analyst DP considers data analysts who are obliged under laws/regulations should not share their budget/query responses to each other (e.g., internal data analysts should not share their results with an external third party). Our mechanism ensures that (i) even if these data analysts break the law and collude, the overall privacy loss is still minimized; (ii) if they do not collude, each of the analysts $A_i$ has a privacy loss bounded by $\epsilon_i$ (c.f. our multi-analyst DP, Definition~\ref{def:multi_analyst_dp}). However,~\cite{pujol2022multi} releases more information than the specified $\epsilon_i$ for analyst $A_i$ (as~\cite{pujol2022multi} guarantees DP).
\ifarxiv
Some other DP systems, e.g. Sage \cite{LecuyerSVG019sage} and its follow-ups \cite{LuoPTCGL21scheduling,Pierre2022packing,nicolas2023cohere}, involve multiple applications or end-users, and they care about the budget consumption \cite{LecuyerSVG019sage,nicolas2023cohere} or fairness constraints \cite{LuoPTCGL21scheduling} in such scenarios. Their design objective is orthogonal to ours --- they would like to avoid running out of privacy budget and maximize utility through batched execution over a growing database.
\fi

\revise{\linelabel{line:related_work2}
The idea of adding correlated Gaussian noise has been exploited in existing work \cite{BaoYXD21CGM,LiCLZ12multilevel}.
However, they all solve a simpler version of our problem. 
Li et al. \cite{LiCLZ12multilevel} resort algorithms to release the perturbed results to different users \textit{once} and Bao et al. \cite{BaoYXD21CGM} study the sequential data collection of \textit{one user}. When putting the two dimensions together, \textit{understudied} questions, such as how to properly answer queries to an analyst that the answer to the same query with lower accuracy requirements has already been released to another analyst, arise and are not considered by existing work. Therefore, we propose the provenance-based additive Gaussian approach (Section \ref{subsec:additive_gaussian_mechanism}) to solve these challenges, which is not merely injecting correlated noise into a sequence of queries.
}

\section{Conclusion and Future Work}
\label{sec:conclusion}

We study how the query meta-data or provenance information can assist query processing in multi-analyst DP settings.
We developed a privacy provenance framework for answering online queries, which tracks the privacy loss to each data analyst in the system using a provenance table.
Based on the \provT, we proposed DP mechanisms that leverage the provenance information for noise calibration and built \oursystem to maximize the number of queries that can be answered.
\oursystem can serve as a middle-ware to provide a multi-analyst interface for existing DP query answering systems.
We implemented \oursystem 
and our evaluation shows that \oursystem can significantly improve the number of queries that could be answered and fairness, compared to baseline systems.

\ifarxiv
\revise{
While as an initial work, this paper considers a relatively restricted but popular setting in DP literature, we believe our work may open a new research direction of using  provenance information for multi-analyst DP query processing. 
We thereby discuss our ongoing work on \oursystem and envision some potential thrusts in this research area.
\squishlist
\item \textbf{Tight privacy analysis.} In the future, we would like to tighten the privacy analysis when composing the privacy loss of the local synopses generated from correlated global synopses.

\item \textbf{Optimal processing for highly-sensitive queries.} While currently \oursystem can be extended to answer these queries na\"ively (by a truncated and fixed sensitivity bound), instance optimal processing of these queries \cite{dong2022r2t} requires data-dependent algorithms, which is not supported by our current approaches. Our ongoing work includes enabling \oursystem with the ability to optimally answer these queries.

\item \textbf{System utility optimization.} We can also optimize the system utility further by considering a more careful design of the structure of the cached synopses \cite{mazmudar2022cache,Ghazi23mutlti-table}, e.g. cumulative histogram views, making use of the sparsity in the data itself \cite{xiao2010differential}, or using data-dependent views \cite{li2014data}.

\item \linelabel{line:DP_future} \textbf{Other DP settings.} \oursystem considers minimizing the collusion among analysts over time in an online system. The current design enforces approximate DP due to the nature of Gaussian properties. Our ongoing work extends to use Renyi DP or zCDP for privacy composition.
Future work can also consider other noise distributions, e.g. Skellam~\cite{bao2022skellam}, to support different DP variants, or other utility metrics\linelabel{line:future_utility}, e.g., confidence intervals~\cite{ge2019apex} or relative errors, for accuracy-privacy translation\ifarxiv
, or other application domain, e.g., location privacy \cite{yu2022thwarting,he2015dpt}
\fi
. 

\item \linelabel{line:open_problem_ac} \textbf{Other problems with access control settings.} More research questions and works may be spawned by a deeper intertwinement between privacy provenance and access/leakage control \cite{Pappachan2022tattletale}.
For instance, an analyst can submit queries with different privilege levels, which requires a more expressive model for privacy provenance. Another use case is that the user can temporarily delegate his/her privilege to others \cite{WangLC08delegation1,TheimerNT92delegation2}.
One potential approach to resolving the first scenario may be to enable hierarchical privacy accounting over a more fine-grained provenance table, i.e., composite privacy loss as per each privilege level and per analyst, and the privacy loss per each data analyst is still bounded in a similar way.
For the second use case, the system may further allow a ``grant'' operator that records more provenance information -- for example, the privacy budget consumed by a lower-privileged analyst during delegation is accounted to the analyst who grants this delegation.

\squishend
}
\fi

\begin{acks}
This work was supported by NSERC through a Discovery Grant. 
We would like to thank the anonymous reviewers for their detailed comments which helped
to improve the paper during the revision process.
We also thank Runchao Jiang, Semih Salihoğlu, Florian Kerschbaum, Jiayi Chen, Shuran Zheng for helpful conversations or feedback at the early stage of this project.

\end{acks}

\bibliographystyle{ACM-Reference-Format}
\bibliography{contents/ref.bib}

\ifarxiv
\clearpage
\appendix
\section{Other DP Notions}

\ifarxiv
\begin{table}[t]
\caption{Notation Summary}
\label{tab:notation}
\begin{tabular*}{\columnwidth}{r|l}
\toprule \hline
{\em Notation} & {\em Definition}  \\
 \hline \\[-1em]
 $\epsilon, \delta$ & The privacy budget \\
 $\mc{A}, A_i$ & (A set of) Data analyst(s) \\
 $\databasedom, D$ & Database domain, database instance \\
 $q_i, v_i$ & Query and accuracy requirement associated\\
 $\mc{P}$ & The \provT\\
  $\Psi, \psi$ & (A set of) Privacy constraint(s) \\
 $\mc{V}, V$ & (A set of) View(s) \\
 $V^{\epsilon}, V^{\epsilon}_{A_i}$ & Global/Local synopsis \\ \hline

\bottomrule
\end{tabular*}
\end{table}

We summarize the common notations in Table \ref{tab:notation} and some other DP notions that are also included and supported in \oursystem.
\fi

\begin{theorem}[Advanced Composition \cite{kairouz15,KairouzOV17}]
    For any $\epsilon \geq 0$ and $\delta \in (0, 1)$, the $k$-fold composition of $(\epsilon, \delta)$-DP mechanisms satisfies $((k-2i)\epsilon, 1- (1-\delta)^k(1-\delta_i))$-DP, for all $i=\{ 0, 1, \ldots, \lfloor k/2 \rfloor \}$, where $\delta_i = \frac{\sum_{\ell = 0}^{i-1}\binom{k}{\ell}(e^{(k-\ell)\epsilon}-e^{(k-2i+\ell)\epsilon})}{(1+e^{\epsilon})^k}$.
\end{theorem}

\begin{definition}[$(\alpha, \epsilon)$-R\'enyi Differential Privacy \cite{mironov2017renyi}]
    For $\alpha \geq 1$, we say a privacy mechanism $\mc{M}: \mc{D} \rightarrow \mc{O}$ satisfies $(\alpha, \epsilon)$-R\'enyi-differential-privacy (or RDP) if for all pairs of neighbouring databases D and D' and all $O \subseteq \mc{O}$, we have
    $$
    \mathbb{E}_{O \sim \mc{M}(D)}[e^{(\alpha-1)\mc{L}_{D/D'}(O)}]\leq e^{(\alpha-1)\epsilon}.
    $$
\end{definition}

When $\alpha \rightarrow 1$, $(1, \epsilon)$-RDP is defined to be equivalent to $\epsilon$-KL privacy \cite{desfontaines2020sok,van2014renyi}.

\begin{theorem}[R\'enyi DP Composition \cite{mironov2017renyi}]
    Given two mechanisms $\mc{M}_1: \databasedom \rightarrow  \mc{O}_1$ and $\mc{M}_2: \databasedom \rightarrow  \mc{O}_2$, s.t. $\mc{M}_1$ satisfies $(\alpha, \epsilon_1)$-RDP and $\mc{M}_2$ satisfies $(\alpha, \epsilon_2)$-RDP.
The combination of the two mechanisms $\mc{M}_{1, 2}: \databasedom \rightarrow  \mc{O}_1 \times \mc{O}_2$, which is a mapping $\mc{M}_{1, 2}(D) = (\mc{M}_1(D), \mc{M}_2(D))$, is $(\alpha, \epsilon_1+\epsilon_2)$-RDP.
\end{theorem}

\begin{theorem}[R\'enyi DP to $(\epsilon, \delta)$-DP \cite{mironov2017renyi}]
    If a mechanism $\mc{M}$ satisfies $(\alpha, \epsilon)$-RDP, it also satisfies $(\epsilon+\frac{\log 1/\delta}{\alpha -1}, \delta)$-DP, for any $\delta \in (0, 1)$.
\end{theorem}

\oursystem supports advanced composition and R\'enyi DP composition for privacy accounting techniques, other than the basic (sequential) composition.

\section{Omitted Proofs}

\subsection{Proof of Theorem \ref{thm:sequential_composition_multi_dp}}

\begin{proof}[Proof of Theorem \ref{thm:sequential_composition_multi_dp}]
    Note that by our definition of multi-analyst DP, if $\mc{M}_1$ and $\mc{M}_2$ satisfies the notion of multi-analyst DP, then on each coordinate (i.e., for each data analyst), the mechanism provides DP guarantee according to each data analyst's privacy budget.
    Applying the sequential composition theorem (Theorem \ref{thm:dp_seq_composition}) to each coordinate, we can get this composition upper bound for multi-analyst DP.
\end{proof}

\subsection{Proof of Theorem \ref{thm:additiveGM_privacy}}

\begin{proof}[Proof of Theorem \ref{thm:additiveGM_privacy}]
For each data analyst $A_j$, where $j \in [n]$, its corresponding privacy budget is $(\epsilon_j, \delta)$ and the additive Gaussian mechanism $\mathcal{M}_{aGM}$ returns $r_j$ to $A_j$.
We first prove the additive Gaussian mechanism $\mathcal{M}_{aGM}$ satisfies multi-analyst-DP, that is, we need to show that,
\begin{equation}
\label{eq:proof_aGM}
\Pr[\mathcal{M}_{aGM}(D) = r_j] \leq e^{\epsilon_j} \Pr[\mathcal{M}_{aGM}(D') = r_j] + \delta_j. 
\end{equation}

\noindent
Case 1. If $\epsilon_j = \max \{ \epsilon_1, \epsilon_2, \dots, \epsilon_n \}$, it is not hard to see the noise generation and addition to the query answer are the same as in the analytic Gaussian mechanism, and therefore equation \ref{eq:proof_aGM} holds.

\noindent
Case 2. If $\epsilon_j \neq \max \{ \epsilon_1, \epsilon_2, \dots, \epsilon_n \}$, w.l.o.g., we can assume $\epsilon_i = \max \{ \epsilon_1, \epsilon_2, \dots, \epsilon_n \}$.
If we use analytic Gaussian mechanism to calculate the Gaussian variance, we have,
$$
\Phi_{\mc{N}} \left( \frac{\Delta q}{2 \sigma_i} - \frac{\epsilon_i \sigma_i}{\Delta q} \right)
    - e^{\epsilon_i} \Phi_{\mc{N}} \left( - \frac{\Delta q}{2 \sigma_i} - \frac{\epsilon_i \sigma_i}{\Delta q} \right) 
    \leq \delta
$$
and
$$
\Phi_{\mc{N}} \left( \frac{\Delta q}{2 \sigma_j} - \frac{\epsilon_j \sigma_j}{\Delta q} \right)
    - e^{\epsilon_j} \Phi_{\mc{N}} \left( - \frac{\Delta q}{2 \sigma_j} - \frac{\epsilon_j \sigma_j}{\Delta q} \right) 
    \leq \delta.
$$
Then for the noisy answer returned to data analyst $A_j$, we have $r_j \coloneqq r_i + \eta_j = r + \eta_i + \eta_j$, where $\eta_i \sim \mathcal{N}(0, \sigma_i^2)$ and $\eta_j \sim \mathcal{N}(0, \sigma_j^2 - \sigma_i^2)$.
Since $\eta_i$ and $\eta_j$ are independent random variables drawn from Gaussian distribution, we obtain $\eta_i + \eta_j \sim \mathcal{N}(0, \sigma_j^2)$.
Thus, Equation (\ref{eq:proof_aGM}) holds according to the analytic Gaussian mechanism (Def. \ref{def:analytic_gaussian}).

Then we need to prove $\mathcal{M}_{aGM}$ satisfies $(\max \{ \epsilon_1, \epsilon_2, \dots, \epsilon_n \}, \delta)$-DP.
In the additive Gaussian mechanism, we look at the true query answer (and the data) only once and add noise drawn from $\mathcal{N}(0, \sigma_i^2)$ to it (assuming $\epsilon_i = \max \{ \epsilon_1, \epsilon_2, \dots, \epsilon_n \}$).
According to the post-processing theorem of DP (Theorem \ref{thm:dp_post_processing}), the overall privacy guarantee is bounded by $(\max \{ \epsilon_1, \epsilon_2, \dots, \epsilon_n \}, \delta)$-DP. 
\end{proof}

\subsection{Proof of Theorem \ref{thm:privacy_guarantee}}

\begin{proof}[Proof of Theorem \ref{thm:privacy_guarantee}]
We analyze the privacy guarantee for the \baseline approach and the additive Gaussian approach separately.

\noindent
\emph{\underline{Analysis of the \baseline approach.}}
In \baseline approach, the accuracy requirement is translated to the minimal privacy cost (as in Def. \ref{def:analytic_gaussian_translation}) and recorded in the \provT $P$ by composition tools if the sanity checking passes.
The $[\ldots, (A_i, \psi_{A_i},\delta),\ldots]$-multi-analyst-DP is guaranteed by the analyst constraint checking on $P$.
The privacy loss on each view is either bounded by its own constraint $\psi_{V_j}$ or bounded by the table constraint $\psi_P$ for water-filling specification (Def. \ref{def:view_constraint_water_filling}), so \oursystem guarantees $\min (\psi_{V_j}, \psi_P)$-DP for view $V_j$.
The table constraint $\psi_P$ bounds the overall information disclosure on the private database. Hence, even when all data analysts $A_1, \ldots, A_m$ collude, the \baseline mechanism is $\psi_P$-DP.

\noindent
\emph{\underline{Analysis of the additive Gaussian approach.}}
Recall that in the additive Gaussian approach, the noise added to different analysts' queries that can be answered on the same view is calibrated using additive GM.
By Theorem \ref{thm:additiveGM_privacy}, for each view $V_j$, the mechanism achieves $[\ldots, (A_i, \min(\psi_{A_i}, \psi_{V_j}),\delta),\ldots]$-multi-analyst-DP and $\psi_{V_j}$-DP for collusion.
Composing for all views, the mechanism provides $[\ldots, (A_i, \psi_{A_i},\delta),\ldots]$-multi-analyst-DP and provides $\min(\max \psi_{A_i}, \sum \psi_{V_j}) = \psi_P$-DP guarantee for the protected DB.
Note that additive Gaussian approach use the same view constraint specification as in \baseline approach, therefore, provides the same per view privacy guarantee.
\end{proof}

\subsection{Proof of Theorem \ref{thm:fairness_guarantee}}

\begin{proof}[Proof of Theorem \ref{thm:fairness_guarantee}]
Recall the proportional fairness says that
for system or mechanism $\mc{M}$, $\forall A_i, A_j$ ($i\neq j$), $l_i \leq l_j$, we have
$
\frac{Err_i(M, A_i, Q)}{\mu(l_i)} \leq \frac{Err_j(M, A_j, Q)}{\mu(l_j)}.
$
Note that $Err_i(M, A_i, Q)$ represents the amount of privacy budgets spent by analyst $A_i$.
When a sufficient number of queries are submitted and answered, we have $\lim_{|Q|\to \infty} Err_i(M, A_i, Q) = \psi_{A_i}$.
Then we show that the  inequality holds for both analyst constraint specifications.
    
\noindent
\emph{\underline{Analysis of the \baseline approach.}}
(Def. \ref{def:vanilla_analyst_constraint})
In \baseline approach, the analyst constraint of $A_i$ is set as $\psi_{A_i} = \frac{l_i}{\sum_{j\in [n]} l_j} \psi_P$.
For a given set of analysts and their privilege levels, $\frac{\psi_P}{\sum_{j\in [n]} l_j}$ is a constant, and therefore $\psi_{A_i}$ is a linear function of $l_i$.
Then for $l_i \leq l_j$, $
\frac{Err_i(M, A_i, Q)}{\mu(l_i)} \leq \frac{Err_j(M, A_j, Q)}{\mu(l_j)}.
$ holds for any linear $\mu(\cdot)$.

\noindent
\emph{\underline{Analysis of the additive Gaussian approach.}}
(Def. \ref{def:agm_analyst_constraint})
In the additive Gaussian approach, the analyst constraint of $A_i$ is set as $\psi_{A_i} = \frac{l_i}{l_{max}}\psi_P$.
Note that $l_{max}$ is the maximum privilege in the system, not the maximum privilege of a set of analysts.
Thus $\frac{\psi_P}{l_{max}}$ is also a constant which makes the inequality hold as well.
\end{proof}

\section{Other Experimental Evaluation}
\label{app_exp}

\subsection{Implementation Details}

We implement \oursystem in Scala 2.12.2 and use sbt as the system dependency management tool.
First, we re-implement the analytic Gaussian mechanism \cite{BW18analytic} library in Scala.
Second, we use PostgreSQL as the underlying database platform and utilize the Breeze library\footnote{Breeze \cite{Breeze} is a numerical processing and optimization library for Scala.} for noise calibration and solving the optimization problem in privacy translation (c.f. additive Gaussian approach).
\oursystem works as a middleware between the data analysts and the existing DP SQL query processing system, and provides advanced functionalities such as privacy provenance, query control, query answering over the cached views, and privacy translation.
\oursystem enables additional mechanisms like additive Gaussian mechanism to allow existing DP SQL systems to answer more queries with the same level of overall privacy protection under the multi-analyst DP framework.
To demonstrate this as a proof-of-concept, we build \oursystem over Chorus \cite{johnson2020chorus}, which is an open-source DP query answering system.
That is, \oursystem involves Chorus as a sub-module and uses its built-in functions, such as database connection and management, query rewriting (i.e., we use Chorus queries to construct the histogram DP synopses), and privacy accountant (including the basic composition and the R\'enyi composition).
Other than the algorithms and necessary data structures supported by \oursystem, we build the view management and \provT management modules, and implement a query transformation module that can transform the incoming queries into linear queries over the DP views/synopses.

\stitle{Implementation of the Provenance Table.}
The \provT as a matrix can be sparse.
Similar to the conventional wisdom applied in access control \cite{samarati2000access}, we can store the \provT by row or column to reduce the storage overhead.

\subsection{Other Results on TPC-H Dataset}

We present the missing empirical results on TPC-H dataset, including the end-to-end comparison (Fig. \ref{fig:tpch_e2e}) and the additive GM component analysis (Fig. \ref{fig:agm_tpch}).

\subsection{Other Results on Adult Dataset}

We present the missing run-time performance comparison on Adult dataset in Table~\ref{tab:tpc_run_time_performance}.

\begin{table}[tbp]
    \centering
    \caption{Runtime Performance Comparison over Different Mechanisms on Adult Dataset (running on a Linux server with 64 GB DDR4 RAM and AMD Ryzen 5 3600 CPU, measured in milliseconds, average over 4 runs)}
    \resizebox{1.0\linewidth}{!}{%
    \begin{tabular}{ccccc}
    \toprule \hline
    Systems  & Setup Time   & Running Time & No. of Queries & Per Query Perf  \\ \hline
    \oursystem     &  8051.06 ms    &  541.77 ms  & 294.5  & 1.84 ms \\ 
    Vanilla &  8051.06 ms    & 393.42 ms  & 289.5   & 1.36 ms  \\ 
    sPrivateSQL &  8051.06 ms &  564.68 ms & 295.8 & 1.91 ms \\ 
    Chorus &  N/A   &  3583.47 ms  & 212.0   & 16.90 ms \\ 
    ChorusP & N/A   &  3424.79 ms  & 212.0  & 16.15 ms \\ \hline
    \bottomrule
\end{tabular}
    }
    \label{tab:run_time_performance}
\end{table}

\begin{figure*}
    \centering
    \includegraphics[width=\linewidth]{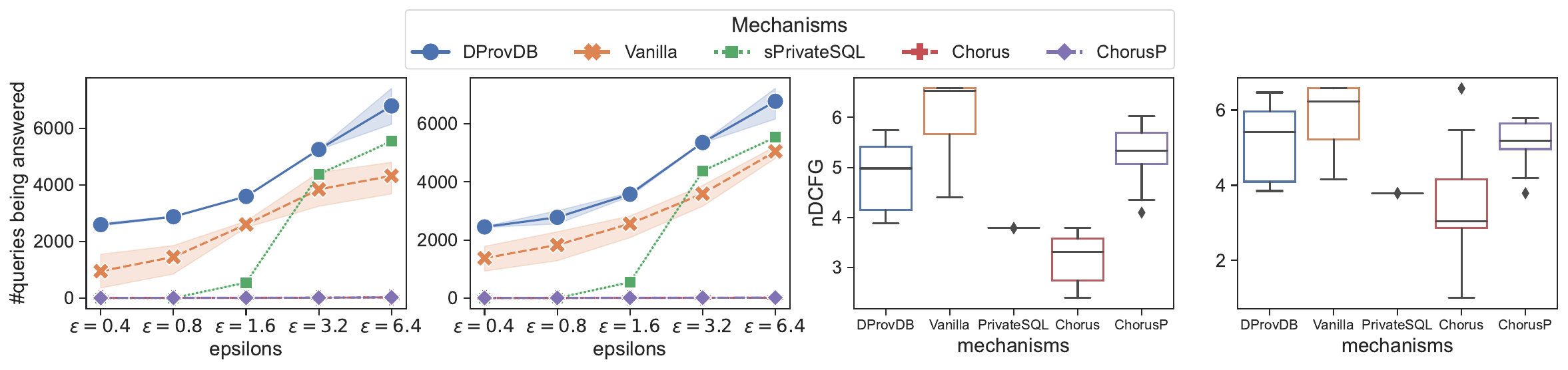}
    \caption{End-to-end Comparison (\emph{RRQ} task, over \emph{Adult} dataset), from left to right: a) utility v.s. overall budget, round-robin; b) utility v.s. overall budget, randomized; c) fairness against baselines, round-robin; d) fairness against baselines, randomized.}
    \label{fig:tpch_e2e}
\end{figure*}

\begin{figure}
    \centering
    \includegraphics[width=\linewidth]{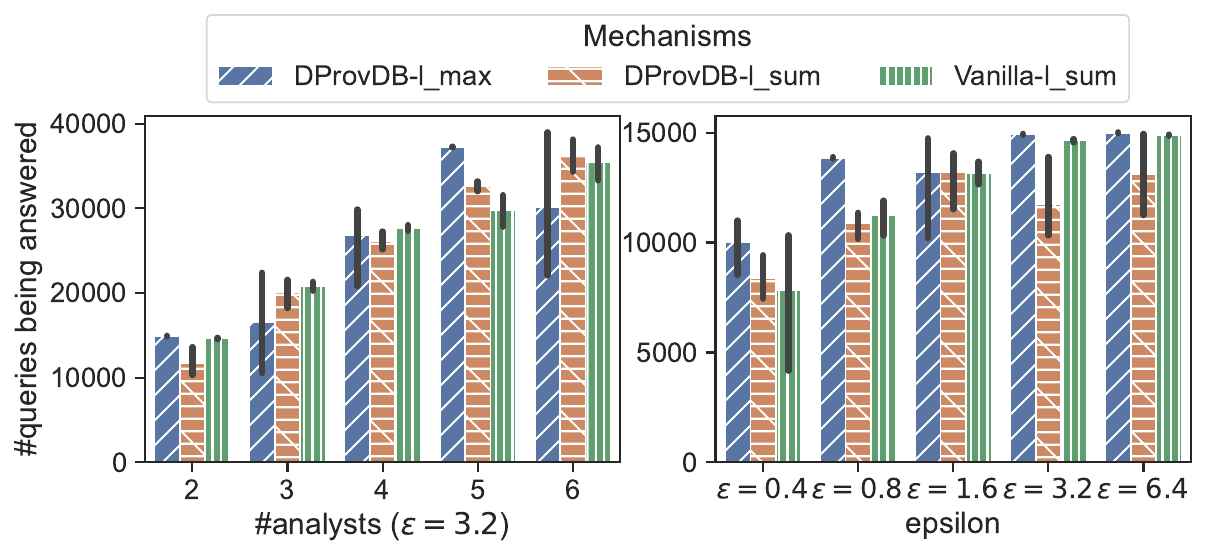}
    \caption{Component Comparison (\emph{RRQ} task, over \emph{Adult} dataset): Additive GM v.s. \Baseline. Left: utility v.s. \#analysts, round-robin; Right: utility v.s. overall budgets, round-robin.}
    \label{fig:agm_tpch}
\end{figure}

\section{Discussion on Design of Views and Supported Queries}

As discussed earlier, \oursystem can properly select a set of materialized histogram views such that incoming queries can be answered using the noisy synopses (corresponding to the views).
Hereby we discuss the design choices of such views for some representative SQL queries.
In particular, we first focus on using full domain histogram to handle GROUP BY operator and then discuss how to handle aggregation queries.

\stitle{Handling the GROUP BY Operator.}
The standard GROUP BY operator in database can leak the information about the activate domain in the database when enforcing differential privacy \cite{mcsherry2009privacy,near2021differential}. Consider the following histogram query that asks the number of employees from each state in a company.

\begin{verbatim}
    SELECT state, COUNT(*)
    FROM employee
    GROUP BY state
\end{verbatim}

By adding noise to the result of this query, the domain knowledge of the State is leaked -- the adversary can clearly know from which states there are no employees of this company.
Existing works \cite{mcsherry2009privacy} modify the GROUP BY operator into GROUP BY* which requires the data analyst to enumerate the domain of the attribute.
This GROUP BY* modification does not directly apply to \oursystem since \oursystem answers queries using DP views.
To enable this operator or this type of histogram query, we explicitly bring up the definition of histogram and full domain histogram. 

\begin{definition}[Full-domain Histogram View]
For a given database instance $D$ and an attribute of the database relation $a_i \in attr(R_i)$, the {histogram} view of this attribute is the 1-way marginal $h(a_i) \in \mathbb{N}^{|Dom(a_i, D)|}$ where each entry of the histogram $h(a_i)[j]$ is the number of elements in the database instance $D$ of value $j \in Dom(a_i, D)$.
The {full-domain histogram} view of this attribute $a_i$ is the histogram view built upon the domain of $a_i$ rather than the active domain.
That is, $h_{f}(a_i) \in \mathbb{N}^{|Dom(a_i)|}$, where $h_f(a_i)[j]$ is the number of elements in the database instance $D$ of value $j \in Dom(a_i)$.
\end{definition}

Note that the non-private full-domain histogram view can be sparse, with a number of 0 entries.
\oursystem generates DP synopses for the full-domain histogram views by injecting independent Gaussian noise to each entry of the view.
Then these DP synopses are able to be extended to handle the differentially private GROUP BY operator by either of the following approaches:
1) in compatible with the GROUP BY* operator, \oursystem can allow data analyst to enumerate the attribute domain and return the corresponding part (or subset) of the DP synopsis to the data analyst; 2) similar to \cite{wilson2020differentially}, \oursystem can allow data analyst to specify a threshold parameter $\tau$, where the noisy entries of the DP synopsis with counts less than $\tau$ will not be returned to the data analyst.

\stitle{Handling Aggregation Queries.}
Aggregation queries requires SQL operators including SUM, AVG, MAX, and MIN.
This type of queries is not supported by PrivateSQL \cite{kotsogiannis2019privatesql}.
A typical running example of aggregation query is to ask the sum of salaries of all employees in the database, which is shown as follows.
The sensitivity of this query equals the maximum value in the domain minus the minimum value in the attribute domain (as we consider the bounded DP).

\begin{verbatim}
    SELECT SUM(salary)
    FROM employee
\end{verbatim}

There are two main challenges for the DP SQL engine to support this type of queries.
First, the DP SQL engine should keep track of the range of the maximum and the minimum value that an attribute can take to correctly calculate the sensitivity of the query.
Second, answering the aggregation queries with the simple (full-domain) histogram view incurs large error variances, since the error accumulates when combining the involved bins in the histogram.

To solve the first issue, we incorporate the clipping and rewriting technique \cite{johnson2020chorus} into the design of the views.
\oursystem allows data curator to enforce a clipping lower bound $lb$ and upper bound $ub$ on the domain of the targeting attribute when generating/initializing the views.
Then the value of the attribute $a_i$ in the database is regarded as $\max (lb, \min (ub, a_i))$.
Therefore, after injecting noise to the histogram (similar to what we described for the GROUP BY operator), the maximum sensitivity of the sum query answered over the histogram is bounded by $(ub - lb)$\footnote{For histogram with domain discretization, the maximum sensitivity is bounded by $\frac{(ub - lb)}{s}$ where $s$ denotes the bin size.}.
We call these generated views as \emph{clipped histogram views}, or clipped views, with bins only in the clipping boundary.
Answering the SUM query over the clipped histogram views is just applying the linear combination of the bin value and the bin counts to (a subset of) the histogram bins.
Since this clipping operation satisfies the property of constant Lipschitz transformation \cite{blocki2013differentially}, generating clipped views satisfy $(\epsilon, \delta)$-DP, derived by applying for the Lipschitz extension on DP \cite{raskhodnikova2015efficient}.

\section{Discussion on Fairness Metrics}

While our system can guarantee proportional fairness with the privacy budget consumed per analyst (when analysts ask a sufficient number of queries) as the quality function, it may be interesting for future work to consider other quality functions such as the average expected error as per workload or the number of queries being answered per analyst.
We coined new fairness metrics (as briefed in experiment section) inspired by the discounted cumulative gain (DCG) in information retrieval \cite{jarvelin2002cumulated,wang2013theoretical}, namely the following discounted cumulative fairness gain (DCFG).
We now provide a bit more information and discussion on these fairness metrics.

\begin{definition}[Discounted Cumulative Fairness Gain]
\label{def:dcfg}
Given $n$ data analysts $A_1, \dots, A_n$, where the privilege level of the $i$-th data analyst is denoted by $p_i$, the fairness scoring of the query-answering of the query engine $\mc{M}$ when the overall privacy budget exhausts is calculated as
$$
\text{DCFG}_\mc{M} = \sum^n_{i=1} \frac{| Q_{A_i} |}{\log_2 (\frac{1}{p_i} +1)}
$$
where $| Q_{A_i} |$ denotes the number of queries answered to the data analyst $A_i$.
\end{definition}

Note that the discounted cumulative fairness gain (DCFG) is proportional to the overall number of queries being answered by the mechanism.
If the number of queries being answered is not comparable to each other for different mechanisms, it is not fair to directly compare and justify this DCFG metric.
Thus a normalized version of DCFG would be useful in practice.

\begin{definition}[Normalized Discounted Cumulative Fairness Gain]
\label{def:ndcfg}
Let $|Q| = \sum^n_{i=1} | Q_{A_i} |$ be the total number of queries answered by the system or mechanism to all data analysts from $A_1, \dots, A_n$:
    The normalized discounted cumulative fairness gain (nDCFG) is calculated as normalizing DCFG by $|Q|$.
    $$
    \text{nDCFG}_\mc{M} = \frac{\text{DCFG}_\mc{M}}{|Q|}
    $$
\end{definition}

\begin{example}
We illustrate the DCFG metric by considering the example where there are three data analysts $A_1, A_2, A_3$ with privilege level $p_1=1, p_2=2, p_3=4$ (i.e., $A_1$ has the lowest privilege level while $A_3$ has the highest privilege level).
Supposing there are two mechanisms $\mc{M}_1, \mc{M}_2$ and the outcomes of them are the following.
$\mc{M}_1$ answers 10 queries to $A_1$, 3 queries to $A_2$ and 0 queries to $A_3$, whereas $\mc{M}_2$ answers 2 queries to $A_1$, 4 queries to $A_2$, and 7 queries to $A_3$.
The DCFG score of the first mechanism is calculated as $\frac{10}{1} + \frac{3}{0.585} + \frac{0}{0.3219} = 15.13$ while that of the second mechanism is $\frac{2}{1} + \frac{4}{0.585} + \frac{7}{0.3219} = 30.58$.
After normalization, the nDCFG scores of these two mechanisms are 1.16 and 2.35, respectively.
Though both mechanisms can answer the same number of queries to the \textit{group of data analysts}, the second one attains higher scores.
\end{example}

\eat{
\section{Yet Another $(t,n)$-Analyst Model}

average-case privacy \cite{barber2014privacy,cuff2016differential} \xh{Is there any average-case analysis in t out of n participating parties in mpc when the n participating parties have different corruption power, which we can borrow. average case of DP never looks the case with different privacy budgets. the adversaries are considered to know different information (\url{https://differentialprivacy.org/average-case-dp/}), but here we refer to adversaries who know the same information but to different extents (budget).}

\stitle{Averaging the compromisation loss.} 

Assuming each analyst $A_i$ is associated with a probability $p_i$ of the extent to which he/she can be compromised.
Under the $(t, n)$-compromised setting, the probabilities of any $t$ out of $n$ analysts sum to less than 1. \xh{Again, can we check in mpc setup, how is t out of n case related to the case with different corruptibility $p_i$?} \xh{$p_i$ is used for privilege level, may use a different term instead}
We would like to bound the average risk of privacy loss in this model.

\begin{definition}[Bounding Average Risk]
Given $n$ analysts with $(t, n)$-compromisation assumption, we say a mechanism satisfies $(\epsilon, t, n)$-DP \xh{the set of $\{p_i\}$ is not modeled in the DP def, I feel it is a different model as the (t,n)-assumption.}, if for all $D, D'$ differing in one tuple,
    $$
    \mathbb{E}_{O \coloneqq (O_1, \dots, O_n) \sim \mc{M}(D)}[p_i \cdot \mathcal{L}_{D/D'}(\mathcal{M}, A_i, O_i)] \leq \epsilon
    $$
    where $\mathcal{L}_{D/D'}(\mathcal{M}, A_i, O_i) = \ln(\frac{\Pr[\mathcal{M}_{A_i}(D)=O_i]}{\Pr[\mathcal{M}_{A_i}(D')=O_i]})$ is the privacy loss random variable of the mechanism's output to analyst $A_i$.
\end{definition}

Note that the privacy notion under the average privacy loss is weaker than DP.
It takes the arithmetic mean of the privacy loss random variable for all data analysts.
The lower bound of the corruption privacy loss, in the average sense, is therefore bounded by $\Sigma_t (p_i \cdot \epsilon_i)$, which is the sum of the $t$ largest $p_i \cdot \epsilon_i$.
Since for any $t$ corruption probability $\Sigma_t p_i \leq 1$, we have $\Sigma_t (p_i \cdot \epsilon_i) \leq \Sigma_t p_i \cdot (\max_t \epsilon_i) \leq \max_t \epsilon_i \leq \max_n \epsilon_i$.
This shows that the lower bound of compromisation under this privacy notion is lower than the bound for all-compromisation setting, indicating an optimal mechanism, if there exists one, can perform better than our proposed algorithm with $(t, n)$-compromisation assumption. \xh{I feel we may not include this averaging the compromisation loss in this version, as it is a quite different notion as (t,n) model, and there is no actual algorithms proposed for this model. Theorem 6.2 seems a good reason to us does not proceed with a new algorithm in this paper. But the purpose of introducing (t,n) is to not waste the privacy budget as discussed in the (t,n)-analysts corruption graph.  }

}
\fi

\balance

\end{document}